\documentclass[12pt]{article}

\usepackage{amssymb,amsmath,amsfonts,amssymb}
\usepackage{graphics,graphicx,color,hhline}
\usepackage{bbm} 
\usepackage{euscript}  
\usepackage{authblk}
 \usepackage{mathtools}
\usepackage{cases}
\def\@abssec#1{\vspace{.05in}\footnotesize \parindent .2in 
{\bf #1. }\ignorespaces} 

\graphicspath{{/EPSF/}{Figures/}}   

\newtheorem{theorem}{Theorem}[section]
\newtheorem{lemma}[theorem]{Lemma}

\newtheorem{remark}[theorem]{Remark}

\def \Rm {\mathbb R}
\def \Nm {\mathbb N}
\def \Cm {\mathbb C}

\newcommand{\eps}{\varepsilon}

\newcommand{\ds}{\displaystyle}

\newcommand{\calC}{\mathcal C}
\newcommand{\calG}{\mathcal G}

\newcommand{\calL}{\mathcal L}

\newcommand{\calI}{\mathcal I}

\newcommand{\calE}{\mathcal E}
\newcommand{\Tr}{\textnormal{Tr}}

\newcommand{\calR}{\mathcal R}

\newcommand{\calJ}{\mathcal J}

\newcommand{\calA}{\mathcal A}
\newcommand{\calM}{\mathcal M}

\newcommand{\calT}{\mathcal T}

\def\fref#1{{\rm (\ref{#1})}}

\newcommand{\cout}[1]{}

\newcommand{\be}{\begin{equation}}
\newcommand{\ee}{\end{equation}}
\newcommand{\bea}{\begin{eqnarray}}
\newcommand{\eea}{\end{eqnarray}}
\newcommand{\bee}{\begin{eqnarray*}}
\newcommand{\eee}{\end{eqnarray*}}
\newcommand{\bal}{\begin{align*}}
                    \newcommand{\eal}{\end{align*}}

\def\Hp{{H}^1_{\textrm{per}}}
\def\H2p{{H}^2_{\textrm{per}}}
\def\Hmp{{H}^{-1}_{\textrm{per}}}

\begin{document}
{\title{On the minimization of quantum entropies under local constraints}}

 \author{Romain  Duboscq \footnote{Romain.Duboscq@math.univ-tlse.fr}}
 \affil{Institut de Math\'ematiques de Toulouse ; UMR5219\\Universit\'e de Toulouse ; CNRS\\INSA, F-31077 Toulouse, France}
 \author{Olivier Pinaud \footnote{pinaud@math.colostate.edu}}
 \affil{Department of Mathematics, Colorado State University\\ Fort Collins CO, 80523}

\maketitle

\begin{abstract}
This work is concerned with the minimization of quantum entropies under local constraints of density, current, and energy. The problem arises in the work of Degond and Ringhofer about the derivation of quantum hydrodynamical models from first principles, and is an adaptation to the quantum setting of the moment closure strategy by entropy minimization encountered in kinetic equations. The main mathematical difficulty is the lack of compactness needed to recover the energy constraint. We circumvent this issue by a monotonicity argument involving energy, temperature and entropy, that is inspired by some thermodynamical considerations.
\end{abstract}

\section{Introduction}
This work is motivated by a series of papers by Degond and Ringhofer about the derivation of quantum hydrodynamical models from first principles. In \cite{DR}, their main idea is to transpose to the quantum setting the entropy closure strategy that Levermore used for kinetic equations \cite{levermore}. Degond and Ringhofer starting point is the quantum Liouville-BGK equation of the form
  \be \label{liouville}
i \hbar \partial_t \varrho=[H, \varrho]+i \hbar Q(\varrho),
\ee
where $\varrho$ is the density operator (a self-adjoint nonnegative trace class operator), $H$ is a given Hamiltonian, $[\cdot,\cdot]$ denotes the commutator between two operators, and $Q$ is BGK-type collision operator that reads, for some relaxation time $\tau$,
$$
Q(\varrho)=\frac{1}{\tau} \left( \varrho_{\rm{eq}}[\varrho]-\varrho \right).
$$
The corner stone of the model is the definition of the quantum statistical equilibrium $\varrho_{\rm{eq}}[\varrho]$, that is obtained by minimizing a quantum entropy under constraints. This is the quantum equivalent
of Levermore's approach of minimization of a classical entropy under constraints of density, velocity, and energy. The well-posedness of the classical minimization problem was addressed in \cite{junk2}, and our main motivation here is to investigate that of the quantum minimization problem. We will focus on the von Neumann entropy,  defined by
$$
S(\varrho)=\Tr \big(s(\varrho)\big), \qquad s(x) = x \log  x -  x,
$$
for a density operator $\varrho$. In \cite{DR}, $\varrho_{\rm{eq}}[\varrho]$ is formally introduced as the unique minimizer of $S$ under \textit{local} constraints of density, current, and energy. The latter are defined as follows (we give further an equivalent definition better suited for the mathematical analysis): to any density operator $\varrho$, we can associate a Wigner function $W[\varrho](x,p)$, see e.g. \cite{LP}, so that the particle density $n[\varrho]$, the current density $n[\varrho] u[\varrho]$, and the energy density $w[\varrho]$ of $\varrho$ are given by similar formulas as in the classical picture,  
$$
\left\{
 \begin{array}{l}
\ds n[\varrho](x)=\int W[\varrho](x,p) dp\\
\ds n[\varrho](x) u[\varrho](x)=\int p W[\varrho](x,p) dp\\
\ds w[\varrho](x)=\frac{1}{2}\int |p|^2 W[\varrho](x,p) dp.
\end{array}
\right.
$$
In the statistical physics terminology, fixing the density, current, and energy amounts to consider equilibria in the microcanonical ensemble. 

Degond and Ringhofer theory raises many interesting mathematical questions, and, to the best of our knowledge, only a few have been answered so far. In \cite{mehats2011problem}, it was proved that the free energy $F_T(\varrho)$, defined (formally) as, for some $T>0$,
$$
F_T(\varrho)= \Tr \big( H \varrho\big) + T \Tr \big( s(\varrho)\big),
$$
where $s$ is the Boltzmann entropy $s(x) = x \log  x -  x$, admits a unique minimizer under the first two local constraints defined on $\Rm^d$, that is $n[\varrho]=n_0$, $u[\varrho]=u_0$, for given functions $n_0(x)$ and $u_0(x)$ with $x \in \Rm^d$ and appropriate regularity assumptions. 

In \cite{MP-JSP}, the minimizer of $F_T$ with the density constraint $n[\varrho]=n_0$ only, was characterized in a one dimensional periodic domain. In the case of the Boltzmann entropy, the minimizer has the form $\varrho=\exp(- (H+A)/T)$ (a ``Quantum Maxwellian''), where $A(x)$ is a function, the so-called chemical potential, i.e. the Lagrange parameter associated with the local density constraint. Under minimal assumptions on $n_0$, it is shown in \cite{MP-JSP} that the potential $A$ belongs to $H^{-1}$ (the Sobolev space), which is sufficient to define the Hamiltonian $H+A$ in the sense of quadratic forms in 1D, but poses problems in higher dimensions. The characterization of the minimizer of $F_T$ as a quantum Maxwellian in dimensions greater than one is then still an open problem. 

In \cite{mehats2017quantum} and still in a one dimensional periodic setting, it was shown that the Liouville-BGK equation \fref{liouville}, with equilibrium defined as the unique minimizer of $F_T$ under the local density constraint only, admits a solution that converges for long time to the unique minimizer of $F_T$ under a \textit{global} density constraint, i.e. the $L^1$ norm of $n[\varrho]$ and not $n[\varrho]$ is prescribed.

The question of existence and uniqueness of a minimizer of $S$ under the constraints $n[\varrho]=n_0$, $u[\varrho]=u_0$, and $w[\varrho]=w_0$, for $n_0$, $u_0$, $w_0$ given functions, has been open since the formal derivations of Degond and Ringhofer. We provide in this paper a first result in a one dimensional setting, and show that $S$ admits a unique minimizer satisfying the constraints. We are limited to 1D since it is not possible at the present time to characterize the minimizer as a quantum Maxwellian in dimensions greater than one. We nevertheless believe our method is sufficiently general to be used in higher dimensions when a proper construction of the quantum Maxwellian is available. Note that the minimization problem can be recast as an inverse problem: given appropriate $n_0$, $u_0$, $w_0$, what is the most likely equilibrium density operator (most likely in the sense that it minimizes a given entropy)  that yields such local density, current, and energy?

The main mathematical difficulty in the minimization of $S$ is to handle the energy constraint. It is indeed direct to show that minimizing sequences $(\varrho_n)_{n \in \Nm}$ satisfy in the limit $n \to \infty$ the density and current constraints, see \cite{mehats2011problem}, but there is not sufficient compactness to recover the energy constraint. More explicitly, we know that $w[\varrho_n]$ converges weakly in $L^1$ some $w[\varrho]$, but we cannot conclude that $w[\varrho]=w_0$. A standard technique such as the concentration-compactness principle \cite{Lions1984} does not seem to apply here since it typically requires some uniform estimates on the sequence $(w[\varrho_n])_{n \in \Nm}$ in addition to the one imposed by the constraint, see e.g. the example provided in \cite[Chapter 4]{evansweak}, while we only know in our problem that the sequence is uniformly bounded in $L^1$. Hence, the recovery of compactness has to come entirely from the properties of the entropy $S$ since no better uniform estimates are available.

Our method goes schematically as follows, and is inspired by the physical principle that, in the canonical ensemble, the (mathematical) entropy decreases with the temperature. Since there is no explicitly defined temperature at this stage in the minimization problem with constraints $n_0$, $u_0$ and $w_0$ (there is a kinetic energy though, so we can expect some implictly defined temperature via the energy), we will artificially introduce it by considering the free energy $F_T$, that we will minimize under the local constraints $n_0$ and $u_0$ for all $T>0$. Having in mind the elementary monotonic relation $E=3/2 k_B T$ between the kinetic energy $E$ and the temperature $T$ of a perfect gas in 3D ($k_B$ is the Boltzmann constant), we will show, and this is the core argument of the proof, that the kinetic energy of the constrained minimizer $\varrho_{T,n_0,u_0}$ of $F_T$ is a strictly increasing function of the temperature. At least from a mathematical point of view, this is by no means a straightforward result since $\varrho_{T,n_0,u_0}$ is only defined via an intricate implicit relation and not explicitly. From the strict increase of the kinetic energy, a standard calculus of variations argument then shows that the entropy of $\varrho_{T,n_0,u_0}$ is strictly decreasing as $T$ increases. This monotonic relation between entropy and temperature subsequently allows us to minimize $S$ under the \textit{local} constraints $n_0$, $u_0$, and a \textit{global} energy constraint. The local energy constraint is then recovered by an argument of the type weak convergence plus convergence of the norm in the space of trace class operators implies strong convergence, together with strict decrease of the entropy.

The article is structured as follows: the functional setting is introduced in Section \ref{prelim}. We work in a one-dimensional periodic domain, and our proofs carries over directly to a bounded 1D domain with Neumann boundary conditions. Our main result and an outline of its proof are given in Section \ref{main}. Section \ref{secproofs} is devoted to the proof the main theorem, and some technical lemmas are stated in an Appendix. \\

\noindent \textbf{Acknowledgment.} OP's work is funded by NSF CAREER grant DMS-1452349.

\section{Preliminaries} \label{prelim} 

\paragraph{Notation.} Our domain $\Omega$ is the 1-torus $[0,1]$. We will denote by $L^r(\Omega)$ , $r\in [1,\infty]$, the usual Lebesgue spaces of complex-valued functions, and by $W^{k,r}(\Omega)$ the standard Sobolev spaces. We introduce as well $H^k(\Omega)=W^{k,2}(\Omega)$, and $(\cdot,\cdot)$ for the Hermitian product on $L^2(\Omega)$ with the convention $(f,g)=\int_\Omega \overline{f} g dx$. We will use the notations $\nabla=d/dx$ and $\Delta=d^2/dx^2$ for brevity. The free Hamiltonian $-\frac{1}{2}\Delta$ is denoted by $H_0$, with domain
$$
\H2p=\left\{u\in H^2(\Omega):\,u(0)=u(1),\,\nabla u(0)=\nabla u(1)\right\}.
$$
We denote by $\Hp$ the space of $H^1(\Omega)$ functions $u$ that satisfy $u(0)=u(1)$. Its dual space is $\Hmp$. Moreover, $\calL(L^2(\Omega))$ is the space of bounded operators on $L^2(\Omega)$, $\calJ_1 \equiv \calJ_1(L^2(\Omega))$ is the space of trace class operators on $L^2(\Omega)$,  and $\calJ_2$ the space of Hilbert-Schmidt operators on $L^2(\Omega)$. With $\Tr(\cdot)$ the operator trace, we will extensively use the facts that, for the duality product $(A,B) \to \Tr (A^* B)$, $\calJ_1$ is the dual of the space of compact operators on $L^2(\Omega)$, and that $\calL(L^2(\Omega))$ is the dual of $\calJ_1$, see \cite{RS-80-I}.
In the sequel, we will refer to a density operator as a nonnegative, trace class, self-adjoint operator on $L^2(\Omega)$. For $|\varrho|=\sqrt{\varrho^* \varrho}$, we introduce the following space:
$$\calE=\left\{\varrho\in \calJ_1:\, \overline{\sqrt{H_0}|\varrho|\sqrt{H_0}}\in \calJ_1\right\},$$
where $\overline{\sqrt{H_0}|\varrho|\sqrt{H_0}}$ denotes the extension of the operator $\sqrt{H_0}|\varrho|\sqrt{H_0}$ to $L^2(\Omega)$. We will drop the extension sign in the sequel to ease notation. The space $\calE$ is a Banach space when endowed with the norm
$$\|\varrho\|_{\calE}=\Tr \big(|\varrho| \big)+\Tr\big(\sqrt{H_0}|\varrho|\sqrt{H_0}\big).$$
Finally, the energy space is the following closed convex subspace of $\calE$:
$$\calE^+=\left\{\varrho\in \calE:\, \varrho\geq 0\right\}.$$

\paragraph{The constraints.} The first three moments of a density operator $\varrho$ are defined as follows: for any smooth function $\varphi$ on $\Omega$, and identifying a function with its associated multiplication operator, the (local) density $n[\varrho]$, current $n[\varrho] u[\varrho]$ and energy $w[\varrho]$ of $\varrho$ are uniquely defined by duality by 

\bee
\int_\Omega n[\varrho] \varphi dx &=& \Tr \big( \varrho \varphi \big)\label{const1}\\
\int_\Omega n[\varrho]u[\varrho] \varphi dx&=& -i \Tr \left(\varrho \left(\varphi \nabla+\frac{1}{2} \nabla \varphi \right) \right) \label{const2}\\
\int_\Omega w[\varrho] \varphi dx&=&  -\frac{1}{2}\Tr \left(\varrho \left( \nabla  \varphi \nabla +\frac{1}{4} \Delta \varphi \right) \right) \label{const3}. 
\eee
We will recast the energy constraint as follows. Denote by $(\rho_p,\phi_p)_{p \in \Nm}$ the spectral elements of a density operator $\varrho$ (eigenvalues counted with multiplicity), and let
$$
k[\varrho]=\frac{1}{2}\sum_{p \in \Nm} \rho_p |\nabla \phi_p|^2.
$$
A short calculation shows that
$$
w[\varrho]=k[\varrho]-\frac{1}{8} \Delta n[\varrho].
$$
Hence, since $n[\varrho]$ is prescribed, we can equivalently set a constraint on $w[\varrho]$ or on $k[\varrho]$, and we choose $k[\varrho]$ since it is nonnegative. Note also the classical (formal) relations
\be \label{defcurrent}
n[\varrho ]= \sum_{p \in \Nm} \rho_p |\phi_p|^2, \qquad  n[\varrho ] u[\varrho ]=\Im \left( \sum_{p \in \Nm} \rho_p \phi_p^* \nabla \phi_p \right).
\ee
For $\varrho \in \calE^+$, we denote the kinetic energy of $\varrho$ by
\begin{equation*}
E(\varrho) = \Tr \big(\sqrt{H_0} \varrho \sqrt{H_0}\big).
\end{equation*}

We introduce below various admissible sets for the minimization problem. The first three impose local constraints on the density, current and energy, while the last one imposes a global constraint on the energy:
\begin{align*}
\calA(n_0) &= \left\{ \varrho\in \calE^+:\; n[\varrho] = n_0 \right\}\\
\calA(n_0,u_0) &= \left\{ \varrho\in \calE^+:\; n[\varrho] = n_0, \; u[\varrho] = u_0\right\}\\
\calA(n_0,u_0,k_0) &= \left\{ \varrho\in\calE^+:\;  n[\varrho] = n_0, \; u[\varrho] = u_0, \; k[\varrho] = k_0\right\}\\
\calA_{g}(e_0) &= \left\{ \varrho\in\calE^+: \; E(\varrho)= e_0\right\}.
\end{align*}

\begin{remark} \label{rem1}We will use the following equivalence: let $\varrho$ be a density operator, then $E(\varrho)<\infty$ (and therefore $\varrho \in \calE^+$) if and only if $k[\varrho] \in L^1$ (in the sense that the series defining $k[\varrho]$ is absolutely convergent in $L^1$). This follows for instance from the observation that $E(\varrho)<\infty$ if and only if $\sqrt{H_0} \sqrt{\varrho} \in \calJ_2$ and from \cite[Theorem 6.22, item (g)]{RS-80-I}. We then have, for all $\varrho \in \calE^+$,
\be \label{relEk}
E(\varrho)=\|k[\varrho]\|_{L^1}.
\ee
\end{remark}

\begin{remark} \label{rem2} A simple calculation shows that, for any $\varrho \in \calE^+$, we have $2 k[\varrho]=-n[\nabla \varrho \nabla]$.
\end{remark}

Throughout the paper, $C$ will denote a generic constant that might differ from line to line.

\section{Main result} \label{main}

We state in this section our main result and give an outline of the proof. We introduce first the set of admissible constraints $\calM$ as
$$
\calM= \left\{ (n_0,u_0,k_0) \in \left(L^1(\Omega) \right)^3: \; n_0=n[\varrho],\, u_0=u[\varrho],\, k_0=k[\varrho],\, \textrm{for some } \varrho \in \calE^+ \right\}.
$$
The set $\calM$ is just the set of constraints $(n_0,u_0,k_0)$ that are the moments of at least one density operator in $\calE^+$. Its characterization in the kinetic case was adressed in \cite{junk} and is not completely straightforward. The question is still open in the quantum case and we expect the problem to be significantly harder than in the classical case. It is nevertheless direct to construct many $(n_0,u_0,k_0)$ in $\calM$: for any $V \in L^\infty(\Omega)$ real-valued, consider for instance $\sigma=f(H_0+V)$ for $f$ smooth, positive, and decreasing sufficiently fast at the infinity (e.g. $\int_0^\infty x^2 f(x) dx <\infty$), and equip $H_0+V$ with the domain $\H2p$; then, for $g(x)=\int_0^x u_0(y)dy, $ the moments of $\varrho=e^{ig }\sigma e^{-ig}$ are in $\calM$. Moreover, the density $n[\varrho]$ is in $\H2p$, $u[\varrho] \in L^2(\Omega)$, $k[\varrho] \in L^1(\Omega)$, $k[\varrho] \geq 0$, and the Krein-Rutman theorem shows that the ground state is strictly positive and as a consequence that $n[\varrho]$ is bounded below. 

Recalling that $S(\varrho)=\Tr(s(\varrho))$ with $s(x) = x \log  x -  x$, our main result is the 
\begin{theorem} \label{thmain} Suppose that $(n_0,u_0,k_0) \in \calM$, where $n_0 \in \H2p$ with $n_0>0$, $u_0 \in L^2(\Omega)$, and $k_0 \in L^1(\Omega)$. Then, the constrained minimization problem
  $$
  \min_{\calA(n_0,u_0,k_0)} S(\varrho)
  $$
  admits a unique solution.
\end{theorem}

The form of the minimizer is characterized formally in \cite{QHD-CMS}, and it reads
$$
\varrho=\exp(-H(A,B,C)),
$$
where, for some functions $(A,B,C)$ called respectively the generalized chemical potential, the generalized mean velocity, and the generalized temperature, we have
$$
H(A,B,C)=- \left(\nabla \cdot \left( \frac{1}{2C} \nabla \right)+\frac{1}{4} \Delta \left( \frac{1}{C}\right)\right)+i \left(\frac{B}{C} \cdot \nabla + \frac{1}{2} \left(\nabla \cdot \frac{B}{C} \right) \right)+A+\frac{|B|^2}{2C}.
$$
When the energy constraint is global, i.e. when $\| k[\varrho]\|_{L^1}$ is prescribed, then the Lagrange parameter $C$ is constant and the Hamiltonian takes the more familiar form
$$
H(A,B,C)=\frac{1}{2C}\left(i\nabla - B\right)^2 + A.
$$
The rigorous justification of the formulas above will be done in a future work. 

We only considered here the Boltzmann entropy since it encodes the main difficulties of the proof. The analysis carries over directly to the Fermi-Dirac entropy of the form $s(x)=x \log x+(1-x) \log (1-x)$, $x\in [0,1]$. More generally, the main properties of the entropy required for Theorem \ref{thmain} to hold true are the following, supposing that $\|n_0\|_{L^1} \leq 1$ for simplicity: $s \in \calC^0([0,1]) \cap \calC^1((0,1))$, $s$ is strictly convex, $(s')^{-1}$ is a strictly positive and strictly decreasing function (this is crucial for the monoticity of the kinetic energy) with $\int_0^\infty x^2 (s')^{-1}(x) dx <\infty$, $S(\varrho)$ is bounded below in $\calE^+$ and continuous for the weak-$*$ topology of $\calE^+$.

\paragraph{Outline of the proof.} We proceed as explained in the introduction. For $(\varrho_n)_{n \in \Nm}$ a minimizing sequence in $\calA(n_0,u_0,k_0)$, the main difficulty is to show that $k[\varrho_n]$ converges weakly in $L^1$ to $k_0$. There are, for this, several steps. The first one is to introduce the free energy $F_T$, defined by, for any $\varrho \in \calE^+$ and any $T>0$,
  $$F_T(\varrho)= E(\varrho)+T S(\varrho).
  $$
  We minimize $F_T$ under two local constraints $n_0$ and $u_0$, and prove the following theorem:
  \begin{theorem} \label{th2mom}Suppose that $n_0 \in \Hp$ with $n_0>0$, and that $u_0 \in L^2$. Then, for any $T>0$, the constrained minimization problem
  \be \label{pbm:mindenscour} 
  \min_{\calA(n_0,u_0)} F_T(\varrho)
  \ee
  admits a unique solution that reads
  $$
  \varrho_{T,n_0,u_0} =  e^{if} \varrho_{T,n_0} e^{-if}, \qquad f(x)=\int_0^x u_0(x) dx,
$$
where $\varrho_{T,n_0}$ is the unique solution to the problem
\begin{equation}\label{prob:densconst}
\min_{\calA(n_0)} F_T(\varrho).
\end{equation}
\end{theorem}

The proof of Theorem \ref{th2mom}, given in Section \ref{proofth2mom}, follows from a simple change of gauge and the fact that \fref{prob:densconst} is well-posed. The core to the proof of Theorem \ref{thmain}, and the most difficult part, is Theorem \ref{thm:NRJIncr} below.

\begin{theorem} \label{thm:NRJIncr} For any $T>0$, let $\varrho_{T,n_0,u_0}$ be the minimizer of Theorem \ref{th2mom} with the additional condition that $n_0 \in \H2p$, and  consider
$$
\mathsf{E}_T=  E(\varrho_{T,n_0,u_0}), \qquad \mathsf{S}_T=S(\varrho_{T,n_0,u_0}).
$$
Then, $\mathsf{E}_T$ and $-\mathsf{S}_T$ are strictly increasing continuous functions on $\Rm_+^*$ and
\be \label{limET}
\lim_{T \to 0^+} \mathsf{E}_T=\frac{1}{2}\left(\int_{\Omega} |\nabla \sqrt{n_0}|^2 dx+\int_{\Omega} n_0 |u_0|^2 dx \right):=m_0.
\ee
\end{theorem}

The proof of Theorem \ref{thm:NRJIncr}, given in Section \ref{proofInc}, is based on a calculation of $\partial_T \mathsf{E}_T$ that raises various difficulties. First, the minimizer $\varrho_{T,n_0}$ in \fref{prob:densconst} admits an implicit representation of the form $\varrho_{T,n_0}=\exp(-(H_0+A_T)/T)$, where $A_T$ is a function that depends on $T$ and on $\varrho_{T,n_0}$ via a complex relation. When differentiating $\mathsf{E}_T$, one needs to differentiate $\varrho_{T,n_0}$ and therefore $A_T$. Due to the intricacy of the relation between $\varrho_{T,n_0}$ and $A_T$, it is difficult to apply the implicit function theorem in order to obtain the differentiability in the variable $T$. We therefore proceed by approximation, and consider a penalized version of the problem \fref{prob:densconst} in which the relation between the minimizer and the corresponding potential $A_T$ is straightforward. It is then possible to the use the implicit function theorem to justify the differentiation with respect to $T$. The price to pay is the need for non-trivial uniform estimates in the penalization parameter to pass to the limit. The second difficulty is the calculation of $\partial_T \mathsf{E}_T$ per se for the penalized problem:  we will see that each terms need to be carefully expressed in order to take advantage of some compensations and show the positivity of the derivative. The last difficulty is to show the strict increase of $\mathsf{E}_T$ for the nonpenalized problem: while there is a strict inequality in the penalized problem, it does not pass to the limit. We fix this by using some results from calculus of variations that allow us to relate $\partial_T \mathsf{E}_T$ and $\partial_T \mathsf{S}_T$, which, together with the uniqueness of solutions to \fref{prob:densconst}, yield the strict increase by a contradiction argument.

With Theorem \ref{thm:NRJIncr} at hand, we can then show in Theorem \ref{th3momg} below that the minimization problem with two local constraints $n_0$, $u_0$ and a global constraint $E(\varrho)=e_0$ admits a unique solution, which is of the form $\varrho_{T_0,n_0,u_0}$, for some implicitly defined temperature $T_0$ that depends on $n_0$, $u_0$ and $e_0$.

  \begin{theorem} \label{th3momg}Suppose that $n_0 \in \H2p$ with $n_0>0$, that $u_0 \in L^2(\Omega)$, and that $e_0 \geq m_0$ for the $m_0$ of Theorem \ref{thm:NRJIncr}.  Then, the constrained minimization problem
  \be \label{pbm:mindenscourlocNRJ}
  \min_{\calA(n_0,u_0) \cap \calA_g(e_0)} S(\varrho)
  \ee
  admits a unique solution. When $e_0 > m_0$ and with the notation of Theorem \ref{th2mom}, the solution has the form $\varrho_{T_0,n_0,u_0}$, where $T_0\equiv T_0(n_0,u_0,e_0)$. When $e_0=m_0$, the set $\calA(n_0,u_0) \cap \calA_g(e_0)$ is reduced to the operator $e^{if}|\sqrt{n_0}\rangle \langle \sqrt{n_0}|e^{-if}$.
\end{theorem}

The proof of Theorem \ref{th3momg} is given in Section \ref{proofth3momg}. The proof of Theorem \ref{thmain}, given in Section \ref{proofthmain}, is concluded by setting $e_0=\|k_0\|_{L^1}$ and by exploiting Theorem \ref{th3momg}, with the observation that the $L^1$ norm of the limit of $k[\varrho_n]$ has to be equal to $e_0$ otherwise there is a contradiction with the strict decrease of the entropy. This, with classical arguments about the convergence of positive operators in $\calJ_1$, allows us to conclude that the limit of $k[\varrho_n]$ is $k_0$. Note that the condition $e_0=\|k_0\|_{L^1} \geq m_0$ is automatically verified for $(n_0,u_0,k_0) \in \calM$. We will see indeed that, for $k_0=k[\varrho]$ with $\varrho \in \calE^+$,
\bee
\|k_0\|_{L^1}=E(\varrho) &\geq& \min_{\calA(n_0)} E(\sigma)+\frac{1}{2} \int_\Omega n_0 |u_0|^2 dx\\
&\geq& \frac{1}{2}\left(\int_{\Omega} |\nabla \sqrt{n_0}|^2 dx+\int_{\Omega} n_0 |u_0|^2 dx \right)=m_0.
\eee

The rest of the paper is devoted to the proof of our main theorem.
\section{Proofs} \label{secproofs}

We start with the proof of Theorem \ref{thmain}. 
\subsection{Proof of Theorem \ref{thmain}} \label{proofthmain}

We remark first that the set $\calA(n_0,u_0,k_0)$ is not empty by construction and address the case $\|k_0\|_{L^1}>m_0$ first. Second, \fref{relEk} and item (i) of Lemma \ref{propentropie} show that the entropy $S(\varrho)$ is bounded below by $-C \|k_0\|^{1/2}_{L^1}$ for any $\varrho \in \calA(n_0,u_0,k_0)$. We can therefore consider a minimizing sequence $(\varrho_{n})_{n\in \Nm}\subset \calA(n_0,u_0,k_0)$ such that

\be \label{lim2inf}
\lim_{n \to \infty} S(\varrho_n)=\inf_{\calA(n_0,u_0,k_0)} S.
\ee
The sequence is bounded in $\calE$ since
\be \label{boundminseq}
\| \varrho_n\|_{\calE}= \Tr \left(\varrho_n\right) + \Tr \left(\sqrt{H_0} \varrho_n \sqrt{H_0} \right)= \|n_0\|_{L^1}+\|k_0\|_{L^1}.
\ee
According to Lemma \ref{strongconv} (i), it follows that, up to a subsequence, $(\varrho_{n})_{n \in \Nm}$ converges strongly in $\calJ^1$ to a certain limit $\varrho$. Then, item (ii) of Lemma \ref{propentropie} shows that $S$ is continuous for the weak-$*$ topology on $\calE^+$, that is, combining with \fref{lim2inf},
\be \label{contS}
S(\varrho)=\lim_{n \to \infty} S(\varrho_n)=\inf_{\calA(n_0,u_0,k_0)} S.
\ee
Note in passing that the continuity of $S$ for the minimizing sequences, and not just the semi-lower continuity, is an important ingredient of the proof. The true question following \fref{contS} is therefore whether $\varrho \in \calA(n_0,u_0,k_0)$ or not. This is clear for the first two moments $n[\varrho]$ and $u[\varrho]$ as there is sufficient compactness, see e.g. \cite[Theorem 2.1 and Theorem 4.3]{mehats2011problem}, and we have
\begin{equation*}
n[\varrho] = n_0 \quad\mbox{ and }\quad u[\varrho] = u_0.
\end{equation*}
The main difficult is therefore to show that $k[\varrho] = k_0$. For this, the key ingredients are Theorems \ref{thm:NRJIncr} and \ref{th3momg}. The starting point is the fact that, since $\sqrt{H_0}\varrho_n\sqrt{H_0}$ is positive and uniformly bounded in $\calJ^1$ according to \fref{boundminseq}, 
\begin{equation} \label{weakJ}
\sqrt{H_0}\varrho_n\sqrt{H_0} \overset{}{\longrightarrow} \sqrt{H_0}\varrho\sqrt{H_0}, \quad\mbox{weak-$*$ in }\calJ^1,
\end{equation}
which is not enough to obtain  $k[\varrho] = k_0$, we would need for this weak convergence in $\calJ_1$ and not weak-$*$ convergence. We will actually prove that \fref{weakJ} holds strongly in $\calJ_1$. To this end, we have first from \fref{relEk} and \fref{weakJ}, together with Lemma \ref{strongconv},
\begin{equation*}
\|k[\varrho]\|_{L^1} \leq \liminf_{n\rightarrow\infty} \|k[\varrho_n]\|_{L^1} = \|k_0\|_{L^1}.
\end{equation*}
If $\|k[\varrho]\|_{L^1} =  \|k_0\|_{L^1}$, we claim that we are done. In order to prove this, we will need the following result:
\begin{lemma}[Theorem 2.21 and addendum H of \cite{Simon-trace}]
\label{theoSimon}
Suppose that $A_k\to A$ weakly in the sense of operators and that $\|A_k\|_{\calJ_1}\to \|A\|_{\calJ_1}$. Then $\|A_k-A\|_{\calJ_1}\to 0$.
\end{lemma}
Then, since by hypothesis
\begin{align*}
 \left\|\sqrt{H_0}\varrho_n\sqrt{H_0}\right\|_{\calJ^1} =  \|k[\varrho_n]\|_{L^1} = \|k_0\|_{L^1} = \|k[\varrho]\|_{L^1} =  \left\|\sqrt{H_0}\varrho\sqrt{H_0}\right\|_{\calJ^1},
\end{align*}
and \fref{weakJ} holds, it follows from Lemma \ref{theoSimon} that (we use here the fact that weak-$*$ convergence in $\calJ_1$ implies the weak convergence in the sense of operators), 
\begin{equation} \label{strongJ1}
\sqrt{H_0}\varrho_n\sqrt{H_0}\longrightarrow \sqrt{H_0}\varrho\sqrt{H_0}, \quad\mbox{strongly in }\calJ^1.
\end{equation}
Then, according to Remark \ref{rem2},
\begin{equation} \label{convk}
2 \|k_0 -  k[\varrho]\|_{L^1} =  2 \|k[\varrho_n] -  k[\varrho]\|_{L^1} = \|\nabla \varrho_n\nabla - \nabla\varrho\nabla\|_{\calJ^1}.
\end{equation}
Finally, it is not difficult to conclude from \fref{strongJ1} and the strong convergence of $\varrho_n$ to $\varrho$ in $\calJ_1$, that the last term in \fref{convk} converges to zero as $n \to \infty$. Hence, we have obtained that $k_0=k[\varrho]$ if $\|k[\varrho]\|_{L^1} =  \|k_0\|_{L^1}$. We therefore assume from now on that
\begin{equation*}
e_1 := \|k[\varrho]\|_{L^1} < \|k_0\|_{L^1} = : e_0,
\end{equation*}
and will prove a contradiction. To this end, since we assumed here that $\|k_0\|_{L^1}>m_0$, it follows from Theorem \ref{th3momg}, for $j\in\{0,1\}$, that there exist $T_j>0$, and a unique $\varrho_j : = \varrho_{T_j,n_0,u_0}$, such that
\begin{equation*}
\varrho_j  = \underset{\calA(n_0,u_0) \cap  \calA_g(e_j)}{\textrm{argmin}} S(\sigma).
\end{equation*}
Therefore, by construction,
$$
e_1 = E(\varrho_1) = \mathsf{E}_{T_1} < e_0 = E(\varrho_0) = \mathsf{E}_{T_0}.
$$
Theorem \ref{thm:NRJIncr} then implies that $T_1<T_0$, and also that $S(\varrho_1)>S(\varrho_0)$.
Since $e_1 = \|k[\varrho]\|_{L^1}$, we have that $\varrho\in\calA(n_0,u_0) \cap  \calA_g(e_1)$. Together with \fref{contS}, this yields
\begin{equation*}
S(\varrho_1) \leq S(\varrho) = \inf_{\calA(n_0,u_0,k_0)} S(\sigma).
\end{equation*}
We will see that there is a contradiction above: by taking the infimum over nonnegative functions $k_0\in L^1$ (we denote this set by $L^1_+$) such that $\|k_0\|_{L^1} = e_0$, and by using Lemma \ref{lem:infmineq} below, we find
\begin{equation*}
\underset{\calA(n_0,u_0) \cap  \calA_g(e_1)}{\textrm{min}} S(\sigma) = S(\varrho_1) \leq \inf_{\substack{k_0\in L^1_+\\\|k_0\|_{L^1} = e_0}}\inf_{\calA(n_0,u_0,k_0)} S(\sigma) = \min_{\calA(n_0,u_0)\cap\calA_g(e_0)} S(\sigma) = S(\varrho_0),
\end{equation*}
which contradicts  $S(\varrho_1)>S(\varrho_0)$. Therefore, our starting hypothesis, namely $\|k_0\|_{L^1}>\|k[\varrho]\|_{L^1}$, is false and we must have $\|k_0\|_{L^1}=\|k[\varrho]\|_{L^1}$. As we have seen, this implies that $k_0=k[\varrho]$ and completes the proof of Theorem \ref{thmain} in the case $\|k_0\|_{L^1}>m_0$ provided we establish Lemma \ref{lem:infmineq} below.

\begin{lemma}\label{lem:infmineq}
Let $e_0 >m_0$. We have the equality
\begin{equation*}
 \inf_{\substack{k_0\in L^1_+\\\|k_0\|_{L^1} = e_0}}\inf_{\calA(n_0,u_0,k_0)} S(\sigma) = \min_{\calA(n_0,u_0)\cap\calA_g(e_0)} S(\sigma).
\end{equation*}
\end{lemma}
\begin{proof}
We will prove the two opposite inequalities to obtain the equality and start with the direction $\leq$. First, by Theorem \ref{th3momg}, there exists $\varrho_0 : = \varrho_{T_0,n_0,u_0}$ such that
\begin{equation*}
\varrho_0 = \underset{\calA(n_0,u_0) \cap \calA_g(e_0)}{\textrm{argmin}} S(\varrho).
\end{equation*}
We show below the intuitive result that $\varrho_0$ minimizes $S$ in $\calA(n_0,u_0)$ with the local energy constraint $k[\varrho]=k[\varrho_0]$, namely
\begin{equation} \label{infk0}
\inf_{\calA(n_0,u_0,k[\varrho_0])} S(\sigma)=\min_{\calA(n_0,u_0,k[\varrho_0])} S(\sigma) = S(\varrho_0).
\end{equation}
Indeed, since $\|k[\varrho_0]\|_{L^1}=e_0$ and as a consequence $\calA(n_0,u_0,k[\varrho_0]) \subset \calA(n_0,u_0)\cap \calA_g(e_0)$, we obtain
\begin{equation}
S(\varrho_0) = \min_{\sigma\in \calA(n_0,u_0)\cap\calA_g(e_0)} S(\sigma) \leq \inf_{\sigma\in \calA(n_0,u_0,k[\varrho_0])} S(\sigma)\label{eq:lemeqinfinf1}.
\end{equation}
Moreover, we have clearly $\varrho_0 \in \calA(n_0,u_0,k[\varrho_0])$ and, hence,
\begin{equation}
 \inf_{\sigma\in \calA(n_0,u_0,k[\varrho_0])} S(\sigma) \leq S(\varrho_0).\label{eq:lemeqinfinf2}
\end{equation}
Combining \eqref{eq:lemeqinfinf1} and \eqref{eq:lemeqinfinf2}, we obtain \fref{infk0}. Consider now the mapping $G\, : L^1_+ \to \mathbb{R}$, given by, $\forall k\in L^1_+$, 
\begin{equation*}
G(k) := \inf_{\calA(n_0,u_0,k)} S(\sigma).
\end{equation*}
Using the fact that $\|k[\varrho_0]\|_{L^1} = e_0$ together with \fref{infk0} and the definition of $\varrho_0$, we deduce the inequality
\begin{equation}
 \inf_{\substack{k_0\in L^1_+\\\|k_0\|_{L^1} = e_0}}\inf_{\calA(n_0,u_0,k_0)} S(\sigma) =  \inf_{\substack{k_0\in L^1_+\\\|k_0\|_{L^1} = e_0}} G(k_0) \leq G(k[\varrho_0]) = S(\varrho_0) = \min_{\calA(n_0,u_0)\cap\calA_g(e_0)} S(\sigma).\label{eq:lemeqinfinf3}
\end{equation}
We now prove the reverse inequality to conclude the proof. For $\|k_0\|_{L^1} = e_0$, we have $\calA(n_0,u_0,k_0) \subset \calA(n_0,u_0)\cap \calA_g(e_0)$ and, thus,
\begin{equation*}
\min_{\calA(n_0,u_0)\cap\calA_g(e_0)} S(\sigma) \leq \inf_{\calA(n_0,u_0,k_0)} S(\sigma).
\end{equation*}
By taking the infimum over $k_0$, this leads to the following inequality
\begin{equation*}
\min_{\calA(n_0,u_0)\cap\calA_g(e_0)} S(\sigma) \leq  \inf_{\substack{k_0\in L^1_+\\\|k_0\|_{L^1} = e_0}}\inf_{\calA(n_0,u_0,k_0)} S(\sigma),
\end{equation*}
which, combined with \eqref{eq:lemeqinfinf3}, concludes the proof of the Lemma.
\end{proof}

\medskip

We conclude the proof of Theorem \ref{thmain} with the case $\|k_0\|_{L^1}=m_0$. Let $\sigma \in \calA(n_0,u_0,k_0)$ such that $k_0=k[\sigma]$. According to Theorem \ref{th3momg}, there is only one such $\sigma$, which reads $\sigma=e^{if}|\sqrt{n_0}\rangle \langle \sqrt{n_0}|e^{-if}$, where  $f$ is defined in Theorem \ref{th2mom}. This ends the proof.

\subsection{Proof of Theorem \ref{th2mom}} \label{proofth2mom}
Consider first the minimization problem \fref{prob:densconst} for any fixed $T>0$. According to \cite[Theorem 2.1]{MP-JSP}, the problem admits a unique solution that reads
\begin{equation}\label{eq:solminim1D}
\varrho_{T,n_0}= \exp \left(-\frac{H_A}{T} \right), \qquad H_A=H_0+A, \qquad A \in \Hmp, \qquad A \textrm{ real-valued}.
\end{equation}
Above, the Hamiltonian $H_A$ is defined in the sense of quadratic forms on $\Hp$.  Denoting by $(\rho_p,\phi_p)_{p \in \Nm}$ the spectral elements of $\varrho_{T,n_0}$, we can choose the eigenfunctions $\phi_p$ to be real-valued since the chemical potential $A$ is real. Besides, it is a classical fact, see e.g. \cite[Theorem A.2]{arnold}, that if $\varrho \in \calE^+$, then the series defining $u[\varrho]$ in \fref{defcurrent} converges in $L^2$. Then, by \fref{defcurrent},
\be \label{zerocurrent}
n_0 u[\varrho_{T,n_0}]=0.
\ee
We claim that the unique minimizer of $F_T$ in $\calA(n_0,u_0)$ is
$$
\varrho_{T,u_0, n_0}=e^{if} \varrho_{T,n_0} e^{-if},
$$
for the $f$ defined in the theorem. We prove this in two steps. 
\paragraph{Step 1: $\varrho_{T,u_0, n_0}$ is admissible.} We verify first that $\varrho_{T,u_0, n_0} \in \calE^+$. It is clear by construction that $\varrho_{T,u_0, n_0}$ is a density operator if $\varrho_{T,n_0}$ is one as well. We then check that 
$E(\varrho_{T,u_0, n_0})$ is finite. For this, we remark that if $(\rho_p,\phi_p)_{p \in \Nm}$ are the spectral elements of $\varrho_{T,n_0}$, then $(\rho_p,e^{if }\phi_p)_{p \in \Nm}$ are those of $\varrho_{T,u_0, n_0}$. By Remark \ref{rem1}, it suffices to verify that
$$
\sum_{p \in \Nm} \rho_p \| \nabla (e^{if }\phi_p)\|^2_{L^2} < \infty
$$
to conclude that $E(\varrho_{T,u_0, n_0})<\infty$. This is direct since
\be \label{shiftEC}
\sum_{p \in \Nm} \rho_p \| \nabla (e^{if }\phi_p)\|^2_{L^2} = \sum_{p \in \Nm} \rho_p \| \nabla \phi_p\|^2_{L^2}+\int_{\Omega}  n_0 |u_0|^2 dx,
\ee
which is finite as $\varrho_{T,n_0} \in \calE^+$ and $u_0 \in L^2$, $n_0 \in H^1 \subset L^\infty$. It is also clear that $n[\varrho_{T,u_0,n_0}]=n[\varrho_{T,n_0}]=n_0$, and it remains to address the current contraint. This is straightforward since
\bea \label{abo}
n[\varrho_{T,u_0,n_0}] u[\varrho_{T,u_0,n_0}] &=& \Im \sum_{p \in \Nm} \rho_p e^{-if}\phi_p^* \nabla (\phi_p e^{if})=  n[\varrho_{T,n_0}] u[\varrho_{T,n_0}]+ n[\varrho_{T,n_0}] u_0 \nonumber \\
&=& n_0 u_0,
\eea
thanks to \fref{zerocurrent}. Hence, $\varrho_{T,u_0, n_0}$ belongs to $\calA(n_0,u_0)$. 
\paragraph{Step 2: $\varrho_{T,u_0, n_0}$ is the minimizer.} Take any $\varrho \in \calA(n_0,u_0)$. We can write
$$
\varrho= e^{if} e^{-if} \varrho e^{if} e^{-if}:=e^{if} \sigma e^{-if},
$$
where by construction $\sigma \in \calE^+$, $n[\sigma]=n[\varrho]$, and $u[\sigma]=0$ by a calculation similar as \fref{abo}. Furthermore, following \fref{shiftEC}, we have the relation
\begin{align*}
F_T(\varrho) &= F_T(\sigma) + \frac 1 2 \int_{\Omega} n_0 |u_0|^2dx.
\end{align*}
Hence,
\begin{align*}
\min_{\varrho\in\calA(n_0,u_0)} F_T(\varrho) &= \min_{\sigma\in\calA(n_0,0)} F_T( e^{i f}\sigma e^{-i f}) 
\\ &\geq \min_{\sigma \in \calA(n_0)} F_T(e^{i f}\sigma e^{-i f})
\\ &=  \min_{\sigma \in \calA(n_0)} F_T(\sigma) + \frac 1 2 \int_{\Omega} n_0 |u_0|^2dx.
\end{align*}

Since $\varrho_{T,n_0}$ is the unique minimizer of $F_T$ in $\calA(n_0)$, it follows that
\bee
\min_{\varrho\in\calA(n_0,u_0)} F_T(\varrho) &\geq& F(\varrho_{T,n_0})+\frac 1 2 \int_{\Omega} n_0 |u_0|^2dx=F_T(e^{i f}\varrho_{T,n_0} e^{-i f})\\
&\geq &F_T(\varrho_{T,u_0,n_0}).
\eee
Since $F_T$ has a unique minimizer in $\calA(n_0,u_0)$ and $\varrho_{T,u_0, n_0} \in \calA(n_0,u_0)$, this shows that $\varrho_{T,u_0,n_0}$ is this minimizer and ends the proof.

\subsection{Proof of Theorem \ref{th3momg}} \label{proofth3momg}
We start by fixing the local constraints $n_0 \in \H2p$, with $n_0>0$, $u_0 \in L^2(\Omega)$, and $e_0>m_0$. For any $T>0$, the minimization problem \eqref{pbm:mindenscour} admits a unique solution $\varrho_{T,n_0,u_0}$ thanks to Theorem \ref{th2mom}. Then, it follows from Theorem \ref{thm:NRJIncr} that $E(\varrho_{T,n_0,u_0})$ is a strictly increasing continuous function of $T$ with limit $m_0$ as $T \to 0$. Hence, for any $e_0>m_0$, there exists a unique $T_0>0$ such that
\begin{equation*}
E(\varrho_{T_0,n_0,u_0}) = \mathsf{E}_{T_0} = e_0.
\end{equation*}
We remark that $T_0$ implicitly depends on $e_0$, $n_0$ and $u_0$. By considering the minimization problem \eqref{pbm:mindenscour} with $T = T_0$, we deduce that
\begin{equation*}
\min_{\calA(n_0,u_0)} F_{T_0}(\varrho) =  F_{T_0}(\varrho_{T_0,n_0,u_0}) = e_0 + T_0 S(\varrho_{T_0,n_0,u_0})  .
\end{equation*}
Moreover, we have by construction that $\varrho_{T_0,n_0,u_0}\in \calA(n_0,u_0) \cap \calA_g(e_0)$ and, hence,
\begin{equation*}
F_{T_0}(\varrho_{T_0,n_0,u_0})=\min_{\calA(n_0,u_0)} F_{T_0}(\varrho) \leq \inf_{\calA(n_0,u_0) \cap \calA_g(e_0)} F_{T_0}(\varrho) \leq F_{T_0}(\varrho_{T_0,n_0,u_0}),
\end{equation*}
which leads to the fact that
\begin{equation*}
\underset{\calA(n_0,u_0) \cap \calA_g(e_0)}{\textrm{argmin}} F_{T_0}(\varrho) = \underset{\calA(n_0,u_0)}{\textrm{argmin}}  \;F_{T_0}(\varrho) = \varrho_{T_0,n_0,u_0}.
\end{equation*}
Thus, we obtain
\begin{align*}
\underset{\calA(n_0,u_0) \cap \calA_g(e_0)}{\textrm{argmin}} S(\varrho) = \underset{\calA(n_0,u_0) \cap \calA_g(e_0)}{\textrm{argmin}}\left(e_0 + T_0 S(\varrho)\right) = \underset{\calA(n_0,u_0) \cap \calA_g(e_0)} {\textrm{argmin}}  F_{T_0}(\varrho) = \varrho_{T_0,n_0,u_0},
\end{align*}
which gives a solution to the minimization problem \eqref{pbm:mindenscourlocNRJ}. The uniqueness then follows from the strict convexity of $S$ in $\calE^+$.

We conclude the proof with the case $e_0=m_0$. Proceeding as in Step 2 of the proof of Theorem \ref{th2mom}, we can show that
$$
\min_{\calA(n_0,u_0)} E(\varrho)=\min_{\calA(n_0)} E(\varrho)+\frac{1}{2}\int_{\Omega} n_0 |u_0|^2 dx.
$$ 
Since $E(\cdot)$ and the constraints are linear in $\varrho$, the minima above are unique.
Then, a direct calculation based on the Cauchy-Schwarz inequality shows that, for any $\varrho \in \calE^+$, 
\be \label{lowEn}
\frac{1}{2}\|\nabla \sqrt{n[\varrho]} \|^2_{L^2} \leq E(\varrho).
\ee
Finally, according the hypotheses on $n_0$, we have that $\sqrt{n_0} \in H^1$, so that the operator $|\sqrt{n_0} \rangle \langle \sqrt{n_0} |$ is in $\calA(n_0)$. Together with \fref{lowEn}, this means that
$$
\min_{\calA(n_0)} E(\varrho)=\frac{1}{2}\|\nabla \sqrt{n_0} \|^2_{L^2}.
$$
Therefore, for any $\sigma \in \calA(n_0,u_0)$ with $\|k[\sigma]\|_{L^1}=m_0$, we have
$$
m_0=\|k[\sigma]\|_{L^1}=E(\sigma) \geq \min_{\calA(n_0)} E(\varrho)+\frac{1}{2}\int_{\Omega} n_0 |u_0|^2 dx=m_0.
$$
Since the minimizer of $E(\cdot)$ in $\calA(n_0)$ is unique, there is only one such $\sigma$ that reads, for the $f$ defined in Theorem \ref{th2mom}, $\sigma=e^{if}|\sqrt{n_0}\rangle \langle \sqrt{n_0}|e^{-if}$. This ends the proof of Theorem \ref{th3momg}.

\subsection{Proof of Theorem \ref{thm:NRJIncr}} \label{proofInc}
The proof comprises several steps that we list below:
\begin{enumerate}
\item Continuity of $\mathsf{E}_T$ and $\mathsf{S}_T$.
\item Limit of $\mathsf{E}_T$ as $T\to 0$.
\item Penalization and uniform bounds.
\item Application of the implicit function theorem to obtain the differentiability  with respect to $T$. 
\item Calculation of $\partial_T \mathsf{E}_T$ and a proof that $\mathsf{E}_T$ is nondecreasing.
\item Derivation of a relation between $\partial_T \mathsf{E}_T$ and $\partial_T \mathsf{S}_T$, proof that $\mathsf{S}_T$ is nonincreasing.
\item Strict increase of $\mathsf{E}_T$ and $-\mathsf{S}_T$.
  \end{enumerate}
\paragraph{Step 1: $\mathsf{E}_T$ and $\mathsf{S}_T$ are continuous.} For any $T>0$, let $\varrho_T :=\varrho_{T,u_0,n_0}$ be the unique solution to problem \fref{pbm:mindenscour}, and consider a sequence $\{T_n\}_{n\geq 0} \subset \Rm^*_+$ such that $T_n \to T$ as $n \to \infty$. Without loss of generality, we assume that $T_n\in [T_{-},T_{+}]$, $\forall n\geq0$. We will show first that 
\begin{equation}\label{eq:cont0}
\varrho_{T_n} \underset{n\rightarrow\infty}\longrightarrow \varrho_{T}, \qquad \textrm{strongly in} \quad \calJ_1, 
\end{equation}
and will do so by bounding uniformly $\varrho_{T_n}$ in $\calE^+$. For this, denote by $(\rho_p,\phi_p)_{p \in \Nm}$ the spectral elements of a density operator $\varrho \in \calE^+$, where the eigenvalues $\rho_p$ form a decreasing sequence with limit zero. Then,
\bea \nonumber
\Tr \big( \varrho \log \varrho  \big)&=&\sum_{p \in \Nm} \rho_p \log \rho_p \leq \sum_{ \rho_p \geq 1} \rho_p \log \rho_p \leq  \log \rho_0 \sum_{ \rho_p \geq 1} \rho_p\\ \label{estlog}
&\leq& \| \varrho\|_{\calJ_1} \log \|\varrho\|_{\calL(L^2)} \leq \| \varrho\|_{\calJ_1} \log \|\varrho\|_{\calJ_1}.
\eea
This, together with the facts that $\varrho_{T_n}$ is the minimizer of $F_{T_n}$ in $\calA(n_0,u_0)$ and that $\varrho_{T_+} \in \calA(n_0,u_0)$, leads to
\begin{equation}\label{eq:cont1}
F_{T_n}(\varrho_{T_n})\leq F_{T_n}(\varrho_{T_{+}})\leq E(\varrho_{T_{+}})+ T_+ \|n_0\|_{L^1} \log \|n_0\|_{L^1}.
\end{equation}
Furthermore, we have according to estimate \fref{souslin} of Lemma \ref{propentropie},
\begin{equation}\label{eq:cont2}
 E(\varrho_{T_n}) - T_{+} C E(\varrho_{T_n})^{1/2} \leq E(\varrho_{T_n}) - T_n C E(\varrho_{T_n})^{1/2}\leq E(\varrho_{T_n})+T_n S(\varrho_{T_n})  = F_{T_n}(\varrho_{T_n}).
\end{equation}
Combining \eqref{eq:cont1} and \eqref{eq:cont2}, we deduce that $E(\varrho_{T_n})$ is uniformly bounded. Thus, up to a subsequence and following Lemma \ref{strongconv}, $\varrho_{T_n}$ converges strongly in $\calJ_1$ to a $\sigma \in \calE^+$, which remains to be identified. To do so, we use again Lemmas \ref{strongconv} and \ref{propentropie} to obtain that 
\begin{align*}
E(\sigma) \leq \underset{n\rightarrow\infty}{\lim\inf} E(\varrho_{T_n}) \qquad\mbox{and} \qquad T S(\sigma) = \underset{n\rightarrow\infty}{\lim} T_n S(\varrho_{T_n}).
\end{align*}
The latter result shows that $\mathsf{S}_T$ is continuous. For the continuity of $\mathsf{E}_T$, we deduce from what is above that
\begin{equation*}
F_T(\sigma) \leq \underset{n\rightarrow\infty}{\lim\inf} F_{T_n}(\varrho_{T_n}) \leq \underset{n\rightarrow\infty}{\lim\sup} F_{T_n}(\varrho_{T_n}) \leq \underset{n\rightarrow\infty}{\lim\sup} F_{T_n}(\varrho_T) = F_T(\varrho_T).
\end{equation*}
Above, we used the fact that $\varrho_{T_n}$ is the minimizer of $F_{T_n}$ in $\calA(u_0,n_0)$, which shows that $F_{T_n}(\varrho_{T_n}) \leq F_{T_n}(\varrho_{T})$ since $\varrho_T \in \calA(n_0,u_0)$. Since the minimizer of $F_T$ in $\calA(n_0,u_0)$ is unique, we conclude that $\sigma = \varrho_T$, and that the entire sequence converges, which yields \eqref{eq:cont0}. Finally, the previous string of inequalities implies the continuity of $F_T(\varrho_T)$, which, since we have seen above that $\mathsf{S}_T$ is continuous, yields that $\mathsf{E}_T$ is continuous.

\paragraph{Step 2: The limit $T \to 0$.} The first part is to show that the free energy of a minimizer over $\calA(n_0,u_0)$ converges as $T \to 0$ to the minimum of the kinetic energy over $\calA(n_0,u_0)$. For this, estimate \fref{estlog}, together with item (i) of Lemma \ref{propentropie} and the use of the Young inequality, yields, for some $C_1$ and $C_2$ positive and independent of $\varrho$,
$$
E(\varrho) (1-C_1 T)-C_2 T \leq F_T(\varrho) \leq E(\varrho)+T \| \varrho\|_{\calJ_1} \log \|\varrho\|_{\calJ_1}, \qquad \forall \varrho \in \calE^+.
$$
Then,
$$
\left(\min_{\calA(n_0,u_0)} E(\varrho) \right)(1-C_1 T)-C_2 T \leq \min_{\calA(n_0,u_0)} F_T(\varrho) \leq \left(\min_{\calA(n_0,u_0)} E(\varrho) \right) + T \|n_0\|_{L^1} \log \|n_0\|_{L^1},
$$
and, as a consequence,
$$
\lim_{T \to 0} \min_{\calA(n_0,u_0)} F_T(\varrho)=\min_{\calA(n_0,u_0)} E(\varrho).
$$
Furthermore, if $(T_n)_{n \in \Nm}$ is any decreasing sequence with limit zero and $T_n \leq 1$, we can show as in Step 1 that $(\varrho_{T_n,n_0,u_0})_{n \in \Nm}$ is uniformly bounded in $\calE^+$, which implies in particular that
$$
\lim_{n \to \infty} \min_{\calA(n_0,u_0)} F_{T_n}(\varrho)=\lim_{n \to \infty} F_{T_n}(\varrho_{T_n,n_0,u_0})=\lim_{n \to \infty} E(\varrho_{T_n,n_0,u_0})=\min_{\calA(n_0,u_0)} E(\varrho).
$$
It thus remains to characterize the minimum of the energy over $\calA(n_0,u_0)$, which, as we have seen at the end of the proof of Theorem \ref{th3momg}, is $m_0$. This concludes the proof of \fref{limET}. 
\paragraph{Step 3: The penalized problem.} We remark once again that, according to \fref{shiftEC}, we have
$$
E(\varrho_{T,n_0,u_0})=E(\varrho_{T,n_0})+\frac{1}{2}\int_{\Omega} n_0 |u_0|^2 dx,
$$
which means that we only need to consider $\varrho_{T,n_0}$ to study the monoticity of $E(\varrho_{T,n_0,u_0})$ since $n_0$ and $u_0$ are independent of $T$. Consider now the penalized functional
\begin{equation*}
F_\eps(\varrho):=F_T(\varrho) + \frac{1}{2\varepsilon}\| n[\varrho] - n_0\|_{L^2}^2.
\end{equation*}
The main properties of $F_\eps$ that we need are the following, see \cite{MP-JSP} for the details: first, $F_\eps$ admits a unique minimizer in $\calE^+$ that we denote by $\varrho_\eps$; furthermore, for all $T>0$ and all $n_0 \in \Hp$, $\varrho_\eps$ converges to $\varrho_{T,n_0}$ strongly in $\calE$ as $\eps \to 0$. Also, $\varrho_\eps$ admits the expression
$$
\varrho_\eps= \exp \left( -\frac{H_0+A_\eps}{T}  \right), \qquad A_\eps=\frac{n[\varrho_\eps]-n_0}{\eps} \in \Hp.
$$  
Note that the relation between $A_\eps$ and $\varrho_\eps$ is linear and fairly simple. This is the main reason for introducing the penalized problem, and will allow us to establish the differentiability of $\varrho_\eps$ with respect to $T$. We have also an estimate of the form (this follows from a similar calculation as \fref{eq:cont2}),
\be \label{unifT}
\|\varrho_\eps\|_{\calE} \leq C(1+T^2),
\ee
which also holds for $\varrho_{T,n_0}$. We have finally
\be  \label{nL2}
\| n[\varrho_\eps] - n[\varrho_{T,n_0}]\|_{L^2}^2 \leq C \eps (1+T^2).
\ee

With the purpose of deriving estimates uniform in $\eps$, it is shown in \cite[Proposition 4.1]{mehats2017quantum}, that the chemical potential $A_\eps$ can be written as (accounting for $T$ and the factor $1/2$ in the free Hamiltonian), with  $n_\eps=n[\varrho_\eps]$,
\be \label{expA}
A_\eps= \frac{1}{n_\eps} \left(\frac{1}{4} \Delta n[\varrho_\eps] + \frac{1}{2} n[\nabla \varrho_\eps \nabla ] + T n[\varrho_\eps \log \varrho_\eps] \right).
\ee
The latter expression has to be understood in $\Hmp$ since we only know at this stage that $\varrho_\eps \in \calE^+$ and therefore that $n_\eps \in \Hp$ by \fref{gradnl2}. The main result of this section is the estimate below, that will be crucial when passing to the limit: for all $T \in [T_1,T_2]$,
\be \label{unifA}
\| A_\eps \|_{L^2} \leq C.
\ee
The constant $C$ depends on the interval $[T_1,T_2]$ but not on $T$. In order to prove \fref{unifA}, we write $n[\varrho_\eps]= \eps A_\eps + n_0$ and find from \fref{expA},
\be \label{eqA}
 \int_{\Omega} n_\eps (A_\eps)^2 dx+\frac{\eps}{4} \int_{\Omega} \left( \nabla A_\eps\right)^2dx = \left(\frac{1}{4}\Delta n_0+\frac{1}{2}n[\nabla \varrho_\eps \nabla ] + T n[\varrho_\eps \log \varrho_\eps], A_\eps \right).
\ee
Since $\varrho_\eps$ converges strongly to $\varrho$ in $\calE$, we can conclude from \fref{gradnl2} and a Sobolev embedding that the respective local densities converge strongly in $\calC^0([0,1])$. In particular, with the Gagliardo-Nirenberg estimate
$$
\|u\|_{\calC^0([0,1])} \leq C \|u\|^{1/2}_{L^2} \|u\|^{1/2}_{H^1}, 
$$
we find from \fref{nL2}, \fref{gradnl2} and \fref{unifT}, 
$$
\|n_\eps-n_0\|_{\calC^0([0,1])} \leq C (1+T^2)^{1/4+3/8}\eps^{1/4}.
$$
For $T \leq T_2$, and since $n_0$ is uniformly bounded from below by a strictly positive constant, we have, for $\eps$ sufficiently small,
\be \label{below}
n_\eps(x) \geq C>0, \qquad \forall x \in [0,1],
\ee
where $C$ does not depend on $T$. Besides, for $(\rho_{p,\eps}, \phi_{p,\eps})_{p \in \Nm}$ the spectral elements of $\varrho_\eps$,
$$
\| n[\nabla \varrho_\eps \nabla ] \|_{L^2} \leq \sum_{p \in \Nm} \rho_{p,\eps} \| \nabla \phi_{p,\eps} \|^2_{L^4},
$$
that we need to control uniformly. For this, we realize first that $\varrho_\eps$ is uniformly bounded in $\calE$ since it converges to $\varrho$ in $\calE$. Then, according to Lemma \ref{regmin},
\be \label{HrH}
\Tr\left( H_0 \varrho_\eps H_0 \right) \leq C (1+T^2)\left(1+ (1+\|A_\eps\|^2_{L^2}) \| \varrho_\eps \|_{\calE}\right) \leq C(1+T^2)(1+\|A_\eps\|^2_{L^2}).
\ee
By the Gagliardo-Nirenberg inequality
$$
\|  \varphi\|_{L^4} \leq C \|  \varphi \|_{L^2}^{3/4} \|  \varphi \|_{H^{1}}^{1/4}, 
$$
we find, together with \fref{HrH} and the Cauchy-Schwarz inequality, that
\bea \nonumber
\| n[\nabla \varrho_\eps \nabla ] \|_{L^2} &\leq & C \| \varrho_\eps \|^{3/4}_{\calE} \left(\| \varrho_\eps \|_{\calE}+ \Tr\left( H_0 \varrho_\eps H_0 \right) \right)^{1/4}\\ \label{estnab}
&\leq& C(1+T^2)^{1/4}\left(1+ \|A_\eps\|^2_{L^2} \right)^{1/4}.
\eea
In the same way, 
\bea \label{elog}
\| n[\varrho_\eps \log \varrho_\eps ] \|_{L^2} &\leq& \sum_{p \in \Nm} | \rho_{p,\eps} \log \rho_{p,\eps}|  \|  \phi_{p,\eps} \|^2_{L^4}\\\nonumber
&\leq& \left(\sum_{p \in \Nm} \rho_{p,\eps} |\log \rho_{p,\eps}|^{4/3} \right)^{3/4}  \left(\sum_{p \in \Nm} \rho_{p,\eps}\| \nabla \phi_{p,\eps} \|^2_{L^2} \right)^{1/4}.
\eea
Since $\|\varrho_\eps\|_{\calJ_1} \leq C_0$, we have that $\rho_{p,\eps} \leq C_0 $, $\forall p \in \Nm$. Then, for $x \in [0,C_0]$, there exists $C$ such that $x |\log x|^{4/3} \leq C x^{2/3}$ for instance. According to Lemma \ref{lieb2} for the second inequality below, we find
$$
\sum_{p \in \Nm} \rho_{p,\eps} |\log \rho_{p,\eps}|^{4/3} \leq C \Tr (\varrho_\eps^{2/3}) \leq C \| \varrho_\eps \|^{2/3}_{\calE}.
$$
This shows that $n[\varrho_\eps \log \varrho_\eps ]$ is bounded in $L^2$ independently of $\eps$. We are now ready to conclude. Going back to \fref{eqA}, using \fref{below} together with \fref{estnab}, yields
$$
\|A_\eps \|_{L^2} \leq C \left( 1+\| \Delta n_0 \|_{L^2}+T+(1+T^2)^{1/4} (1+ \|A_\eps\|^2_{L^2})^{1/4}\right).
$$
Since $\Delta n_0 \in L^2$ by the hypotheses, an homogeneity argument for any $T \in [T_1,T_2]$ leads to \fref{unifA}.

\paragraph{Step 4: $\varrho_{\eps}$ is differentiable with respect to $T$.} We prove in this section that the penalized solution $\varrho_\eps$ is differentiable w.r.t. the temperature $T$. Since the parameter $\eps$ plays no role here, we set it to one to ease notation. We will also use the standard shorthand $\beta=1/T$. We start by recalling the implicit function theorem on Banach spaces.

\begin{theorem}
  Let $X,Y, Z$ be three Banach spaces and $f: X\times Y \to Z$ be a continuously Fr\'echet differentiable mapping. If $(x_{0},y_{0})\in X\times Y$, $f(x_{0},y_{0})=0$ and $y\mapsto Df(x_{0},y_{0})(0,y)$ is a Banach space isomorphism from $Y$ onto $Z$, then there exist neighbourhoods $U$ of $x_0$ and $V$ of $y_0$ and a Fr\'echet differentiable function $g : U \to V$ such that $f(x, g(x)) = 0$, and $f(x, y) = 0$ if and only if $y = g(x)$, for all $(x,y)\in U\times V$.
\end{theorem}

Consider  the functional  $\calI$ defined formally by
$$
\calI(\beta,\varrho) : = \varrho - e^{-\beta (H_0+A[\varrho])},\qquad A [\varrho] : = n[\varrho] - n_0.
$$
 By construction, we have $\calI(\beta,\varrho_\eps)=0$ for any $\beta>0$, and we will show that $\calI$ is continuously differentiable on $\Rm^*_+ \times \calE$. We start with the continuity of $\calI$.\\

\noindent \textit{Step 4a: The functional $\calI$ is continuous from $\Rm^*_+ \times \calE$ to $\calE$.} We prove first that $\calI$ is continuous from  $\Rm^*_+ \times \calE$ to $\calL(L^2)$. For this, consider a sequence $(\beta_k,\varrho_k)_{k\in \Nm} \subset [\beta_0,\beta_1] \times \calE$ such that
\begin{equation}\label{eq:seqContin}
(\beta_k,\varrho_k)  \underset{k\rightarrow+\infty}{\longrightarrow} (\beta,\varrho), \; \mbox{in } [\beta_0,\beta_1] \times \calE.
\end{equation}
Estimate \fref{gradnl2} of Lemma \ref{lieb} shows that $A[\varrho_k]$ is uniformly bounded in $H^1$, which, together with a Sobolev embedding, yields
\begin{equation} \label{convAk}
A[\varrho_k]\underset{k\rightarrow+\infty}{\longrightarrow} A[\varrho], \qquad \mbox{strongly in } L^\infty.
\end{equation}
For any $\varrho\in \calE$, define then the self-adjoint operator $H[\varrho] : = H_0 + A[\varrho]$ with domain $\H2p$ (equipped with its usual norm, which is equivalent to $(\|\varphi\|_{L^2}^2 + \|(H_0 + A[\varrho])\varphi\|_{L^2}^2)^{1/2}$ since $A[\varrho] \in L^\infty$). Write then
$$
\beta H[\varrho]- \beta_k H[\varrho_k]= (\beta-\beta_k)H[\varrho]+\beta_k (H[\varrho]-H[\varrho_k]):=a_k+b_k.
$$
From \fref{eq:seqContin} and \fref{convAk}, we have 
$$
\sup_{\|\varphi\|_{\H2p}=1} \left(\| a_k \varphi \|_{L^2} + \| b_k \varphi \|_{L^2}\right)\underset{k\rightarrow+\infty}{\longrightarrow} 0,
$$
which, according to \cite{RS-80-I}, Theorem 8.25, item (b), yields the convergence of $\beta_k H[\varrho_k]$ to $\beta H[\varrho]$ in the norm resolvent sense. In turn, we conclude from \cite[Theorem 8.20, item (a)]{RS-80-I}, that 
\begin{equation} \label{convL}
\left\|e^{-\beta_k H[\varrho_k]} - e^{-\beta H[\varrho]}\right\|_{\calL(L^2)}\underset{k\rightarrow+\infty}{\longrightarrow}0,
\end{equation}
which proves the continuity of $I$ from  $\Rm^*_+ \times \calE$ to $\calL(L^2)$.

We address now the continuity in  $\calE$ using compactness arguments. First, the operator $H[\varrho_k]$ has a compact resolvent and denote by $(\lambda_m[\varrho_k],\phi_m[\varrho_k])_{m\in \Nm}$ its spectral elements. The min-max principle yields the inequality
\begin{equation}\label{eq:bndeigen}
  \gamma_m - \|A[\varrho_k]\|_{L^\infty} \leq \lambda_m[\varrho_k] \leq \gamma_m + \|A[\varrho_k]\|_{L^\infty},
\end{equation}
where the $\gamma_m $ are the eigenvalues of $H_0$, which, counting muliplicities, read $\gamma_0=0$, $\gamma_{2k} =\gamma_{2k-1}= 2 (\pi k)^2$ for $k \geq 1$. We also have the direct estimate, using \fref{convAk},
$$
\| \nabla \phi_m[\varrho_k]\|^2_{L^2} \leq  2\lambda_m[\varrho_k]+2\|A[\varrho_k]\|_{L^\infty} +C \leq C \gamma_m+C.
$$
This then yields, for $\sigma_k=\exp(-\beta_k H[\varrho_k])$,
\be \label{estsigE}
\Tr \big( \sqrt{H_0} \sigma_k \sqrt{H_0}\big)=\sum_{m \in \Nm} e^{- \beta_k \lambda_m[\varrho_k]} \|\nabla \phi_m[\varrho_k]\|^2_{L^2} \leq C\sum_{m \in \Nm} e^{-\beta_0 (C \gamma_m+C)} (1+\gamma_m)  \leq C.
\ee
Since $\sigma_k$ is positive, this shows that $\sigma_k$ is uniformly bounded in $\calE$. Using \fref{HrH}, we obtain from \fref{convAk} and \fref{estsigE} the estimate
$$
\Tr\left(H_0 \sigma_k H_0 \right) \leq C.
$$
This, together with Lemma \ref{strongconv} item (ii) and \fref{convL}, shows that $\sigma_k$ converges strongly in $\calE$ to $e^{-\beta H[\varrho]}$. This proves the continuity from $\Rm^*_+ \times \calE$ to $\calE$. We address now the differentiability.\\

\noindent \textit{Step 4b: Differentiability of  $\calI$.} With the representation
$$
e^{-\beta H[\varrho]}= \frac 1 {2i\pi} \int_{\gamma} e^{-\beta z} (z - H[\varrho])^{-1}dz,
$$
a direct calculation shows that the Fr\'echet derivative of $\calI$ with respect to $\varrho \in \calE$, in the direction $\sigma \in \calE$ and at the point $(\beta,\varrho)$, reads
\begin{equation*}
D\calI[\beta,\varrho](\sigma) = \sigma -\frac 1 {2i\pi} \int_{\gamma} e^{-\beta z }(z - H[\varrho])^{-1} n[\sigma] (z - H[\varrho])^{-1}dz =: \sigma - Z[\beta,\varrho](\sigma).
\end{equation*}
Above, $\gamma$ is a smooth contour in the complex plane, symmetric around the real axis, and defined as follows: for some $M>0$, denote by $B_M$ the ball $\{\varrho \in \calE, \; \|\varrho \|_{\calE} \leq M\}$. Then, by estimate \fref{ninfty}, there is a constant $c_M>0$ such that
$$
\|A[\varrho]\|_{L^\infty} \leq \|n[\varrho]\|_{L^\infty}+\|n_0\|_{L^\infty} \leq C \|\varrho\|_{\calE}+\|n_0\|_{L^\infty} \leq c_M.
$$
For all $\varrho \in B_M$, and denoting by $\textrm{Sp}(H[\varrho])$ the spectrum of $H[\varrho]$ equipped with the domain $\H2p$, we have
$$
\textrm{Sp}(H[\varrho]) \subset [-c_M,\infty).
$$
Fix $r=c_M+1$. The contour $\gamma$ is then written as $\gamma=\gamma_1 \cup \gamma_2 \cup \gamma_3$, where $\gamma_2=\{z \in \Cm, z= r e^{i \theta}, \; \theta \notin (-\pi/4, \pi/4) \}$, where $\gamma_3$ is a broken line starting at $r e^{i \pi/4}$ and going to the infinity with an angle $\theta_0 \in (0, \pi/2)$, and $\gamma_1$ is the symmetric of $\gamma_3$ with respect to the real axis. The contour $\gamma$ is oriented from the positive to the negative imaginary parts. See \cite[page 249]{RS-80-2} for a depiction of a similar contour. By construction of $\gamma$, there exists then $m_M >0$ such that

\be \label{belowresolv}
 \|(z-H[\varrho])^{-1}\|_{\calL(L^2)} \leq m_M, \qquad \forall z \in \gamma, \qquad \forall \varrho \in B_M.
\ee

We will prove that the derivative $D\calI[\beta,\varrho]$ is continuous from $\Rm^*_+\times \calE$ to $\calL(\calE)$ (the space of bounded operator from $\calE$ to $\calE$). For this, we clearly only need to work with $Z[\beta,\varrho]$, and we will proceed as in the proof of continuity of $\calI$: we show first the continuity from $\Rm^*_+ \times \calE$ to $\calL(\calE, \calL(L^2))$, and then prove some uniform estimates and proceed by compactness. A key ingredient is the following lemma:

\begin{lemma} \label{boundZ} There exists $C_{M,\beta_0, \beta_1}>0$ such that, for all $\varrho \in B_M$ and any $\beta \in [\beta_0,\beta_1]$, we have
$$
\Tr\Big( (1+H_0)^{3/4} |Z[\beta,\varrho](\sigma)| (1+H_0)^{3/4}\Big) \leq C_{M,\beta_0,\beta_1} \|n[\sigma]\|_{L^\infty}, \qquad  \forall \sigma \in \calE.
$$ 
\end{lemma}

The proof is postponed to the end of the section. As announced, we start with the continuity from $\Rm_+^* \times \calE$ to $\calL(\calE, \calL(L^2))$. Consider then a sequence $(\beta_k,\varrho_k)_{k\in \Nm} \subset [\beta_0,\beta_1] \times \calE$ as in \fref{eq:seqContin} and such that $\|\varrho_k\|_{\calE}\leq M_0$, for some $M_0>0$ independent of $k$. We choose the constant $M$ in the construction of the contour $\gamma$ such that $M \geq M_0$ and $\|\varrho\|_{\calE}\leq M$ so that we can use the same contour $\gamma$ for both $\varrho$ and $\varrho_k$. We write 
$$
Z[\beta_k,\varrho_k]-Z[\beta,\varrho]=R_{k,1}+R_{k,2},
$$
where
\bee
R_{k,1}(\sigma)&=&\frac 1 {2i\pi} \int_{\gamma} \left(e^{-\beta_k z }-e^{-\beta z } \right)(z - H[\varrho_k])^{-1} n[\sigma] (z - H[\varrho_k])^{-1}dz\\
R_{k,2}(\sigma)&=&\frac 1 {2i\pi} \int_{\gamma} e^{-\beta z } \calR_k[z] dz
\eee
and 
$$
\calR_k[z](\sigma)=(z - H[\varrho_k])^{-1}n[\sigma](z - H[\varrho_k])^{-1}-(z - H[\varrho])^{-1}n[\sigma](z - H[\varrho])^{-1}.
$$
Using $\fref{ninfty}$ to bound $\| n[\sigma]\|_{L^\infty}$ by $\|\sigma\|_\calE$, we find from \fref{belowresolv} that
$$
\|R_{k,1}(\sigma)\|_{\calL(L^2)} \leq C \|\sigma \|_{\calE} \int_{\gamma} \left|e^{-\beta_k z }-e^{-\beta z } \right|dz.
$$
Since $|e^{-\beta z }-e^{-\beta_k z }| \leq e^{-\beta |z| }+e^{-\beta_0 c |z|}$, for all $\beta_k \in [\beta_0,\beta_1]$ and for some constant $c$ that depends on the contour $\gamma$, it follows from dominated convergence that
$$
\|R_{k,1}[\varrho]\|_{\calL(\calE,\calL(L^2))} \underset{k\rightarrow+\infty}{\longrightarrow}0.
$$
Regarding $R_{k,2}$, we have
$$
\|R_{k,2}(\sigma)\|_{\calL(L^2)} \leq C \|\sigma \|_{\calE} \int_{\gamma} \left|e^{-\beta z } \right| \left\|(z - H[\varrho])^{-1}-(z - H[\varrho_k])^{-1}\right\|_{\calL(L^2)} dz.
$$
Since $\beta_k H[\varrho_k]$ converges to $\beta H[\varrho_k]$ in the norm resolvent sense as proved in Step 4a, and as a consequence $H[\varrho_k]$ to $H[\varrho]$, it follows that 
$$
\|(z - H[\varrho])^{-1}-(z - H[\varrho_k])^{-1}\|_{\calL(L^2)} \underset{k\rightarrow+\infty}{\longrightarrow}0, \qquad \forall z \in \gamma.
$$
This, together with \fref{belowresolv} and dominated convergence, shows that  $R_{k,2}$ converges strongly to zero in $\calL(\calE,\calL(L^2))$ as $k \to \infty$. We have therefore obtained the continuity of $Z[\beta,\varrho]$ from $\Rm_+^* \times \calE$ to $\calL(\calE,\calL(L^2))$, that is
\be \label{convLL1}
\|Z[\beta_k,\varrho_k]-Z[\beta,\varrho]\|_{\calL(\calE,\calL(L^2))} \underset{k\rightarrow+\infty}{\longrightarrow}0.
\ee
 
We now prove the continuity from $\Rm_+^* \times \calE$ to $\calL(\calE)$. Lemma \ref{boundZ} implies that, $\forall k \in \Nm$,
$$
\Tr\Big( (I+H_0)^{3/4} |Z[\beta_k,\varrho_k](\sigma)| (1+H_0)^{3/4} \Big) \leq C_{M,\beta_0,\beta_1} \|\sigma\|_{\calE}, \qquad  \forall \sigma \in \calE.
$$
Fix $\sigma \in \calE$. Applying item (ii) of Lemma \ref{strongconv} together with \fref{convLL1}, we can  conclude that $Z[\beta_k,\varrho_k](\sigma)$ converges to $Z[\beta,\varrho](\sigma)$ strongly in $\calE$. This holds for all $\sigma \in \calE$, and provides us with the pointwise convergence of $Z[\beta_k,\varrho_k]$.

The convergence in $\calL(\calE)$ follows from a standard $\eps/3$ argument: denote by $B_\calE$ the unit ball of $\calE$, which is compact for the the weak-$*$ topology. It is metrizable since the space of compact operators on $L^2$ is separable, and denote by $d_\calE$ an associated  metric. From the pointwise convergence,  one can construct by a diagonal argument a subsequence such that $Z[\beta_k,\varrho_k](\sigma) \to Z[\beta,\varrho](\sigma)$ strongly in $\calE$ for all $\sigma$ in a countable dense subset of $B_\calE$ for the metric $d_\calE$. Since we have some equicontinuity from Lemma \ref{boundZ}, i.e.
$$
\| Z[\beta,\varrho](\sigma) \|_{\calE} \leq C_{M,\beta_0,\beta_1} \|n[\sigma]\|_{L^\infty}, \qquad \forall \varrho \in B_M, \; \forall \beta \in [\beta_0,\beta_1],
$$
it suffices to show that the map $\sigma \to n[\sigma]$ is continuous from $B_\calE$ equipped with $d_\calE$ to the space $L^\infty$ to conclude the $\eps/3$ argument. This is direct: take $\sigma \in B_\calE$ and a sequence $\sigma_p \in B_\calE$ such that $d_\calE(\sigma,\sigma_p) \to 0$. The sequence $n[\sigma_p]$ is then uniformy bounded in $H^1$ by \fref{gradnl2}, and as a consequence of classical Sobolev embeddings, there is a subsequence such that $n[\sigma_p]$ converges to some $n$ strongly in $L^\infty$. According to Lemma \ref{strongconv} (i), there is a subsequence such that $\sigma_p$ converges to $\sigma$ strongly in $\calJ_1$, which implies that $n[\sigma_p] \to n[\sigma]$ strongly in $L^1$. Hence, $n=n[\sigma]$, and we have obtained a subsequence such that $n[\sigma_k] \to n[\sigma]$ in $L^\infty$ when $d_\calE(\sigma_k,\sigma)\to 0$. The fact that the entire sequence converges follows from the uniqueness of the limit. As a conclusion of the $\eps/3$ argument, we have just obtained that
$$
\sup_{\sigma \in B_\calE} \|Z[\beta_k,\varrho_k](\sigma)-Z[\beta,\varrho](\sigma)\|_{\calE} \underset{k\rightarrow+\infty}{\longrightarrow}0,
$$
which ends the proof of continuity of $D \calI$.

The proof that $\partial_\beta \calI(\beta,\varrho)= - H[\varrho] e^{-\beta H[\varrho]}$ is continuous from $\Rm_+^* \times \calE$ to $\calE$ is similar to the proof of continuity of $\calI$ and is omitted.  

\noindent \textit{Step 4c: the implicit function theorem.} Recalling that $\varrho_\eps$ is the unique minimizer of $F_\eps$ with temperature $T=1/\beta$, it remains to prove that $\sigma \to D\calI[\beta,\varrho_\eps](\sigma)$ is an isomorphism in $\calE$ in order to apply the implicit function theorem. For this, set the constant $M$ in the definition of the contour $\gamma$ such that $\varrho_\eps \in B_M$, and remark that Lemma \ref{strongconv} (ii) together with Lemma \ref{boundZ} imply that the operator $\sigma \to Z[\beta,\varrho_\eps](\sigma)$ is compact in $\calE$. We can then apply the Fredholm alternative, and we only need to show the injectivity to obtain the invertibility. Suppose then that for some non zero $\sigma \in \calE$,
\be \label{inject}
D\calI[\beta,\varrho_\eps](\sigma)=0.
\ee
Assume for the moment that, with $n[\sigma] \in L^\infty$,
\be \label{negZ}
\Tr\left(Z[\beta,\varrho_\eps](\sigma) n[\sigma]\right) \leq 0, \qquad \forall \sigma \in \calE.
\ee
Then, 
$$
0=\Tr \left( D\calI[\beta,\varrho_\eps](\sigma) n[\sigma] \right)= \|n[\sigma]\|^2_{L^2}- \Tr\left(Z[\beta,\varrho_\eps](\sigma) n[\sigma]\right) \geq \|n[\sigma]\|^2_{L^2}.
$$
This shows that $n[\sigma]=0$, and as a consequence $Z[\beta,\varrho_\eps](\sigma)=0$, and by \fref{inject}, $\sigma=0$. This yields that $\sigma \to D\calI[\beta,\varrho_\eps]$ is an isomorphism in $\calE$. An application of the implicit function theorem then shows that $\varrho_\eps$ is differentiable in a neighborhood of $\beta$. Since $\beta$ is arbitrary and $\varrho_\eps$ is uniquely defined for all $\beta \in \Rm^*_+$, we can finally conclude that $\varrho_\eps \in \calC^1(\Rm^*_+,\calE)$. 

It remains to prove \fref{negZ} and Lemma \ref{boundZ} to conclude this section.

\noindent \textit{Proof of \fref{negZ}}. Using standard complex analysis, one finds that $Z[\beta,\varrho_\eps]$ admits the expression, for all $\sigma \in \calE$ and all $\varphi \in L^2$, 
\begin{equation}\label{eq:decompZ}
Z[\beta,\varrho_\eps](\sigma) \varphi=\sum_{m,k = 0}^{+\infty}  \varsigma_{m,k}  \phi_k \big(\phi_k,n[\sigma]\phi_m\big) \big(\phi_m,\varphi \big),
\end{equation}
where
\begin{equation*}
 \varsigma_{m,k} : = \left\{\begin{array}{ll}
 -\beta e^{-\beta \lambda_k},\,\mbox{ if } \lambda_k= \lambda_m,
 \\ \frac{e^{-\beta \lambda_k} - e^{-\beta \lambda_m}}{\lambda_k - \lambda_m},\,\mbox{ if } \lambda_k\neq \lambda_m,
 \end{array}\right.
\end{equation*}
and $(\lambda_k,\phi_k)_{k \in \Nm}$ are the spectral elements of $H[\varrho_\eps]$.
According to Lemma \ref{lem:bndvarsigma}, the series above is absolutely convergent in $L^2$. Since $Z[\beta,\varrho_\eps](\sigma)$ is trace class and $n[\sigma] \in L^\infty$, then \fref{eq:decompZ} shows that
$$
\Tr\left(Z[\beta,\varrho_\eps](\sigma) n[\sigma]\right)=\sum_{m,k = 0}^{+\infty}  \varsigma_{m,k}[\varrho] \left|\big(\phi_k[\varrho],n[\sigma]\phi_m[\varrho]\big)\right|^2,
$$
which is nonpositive.\\

\noindent \textit{Proof of Lemma  \ref{boundZ}}.  First of all, for all $\sigma \in \calE$, the operator $Z[\beta, \varrho](\sigma)$ is self-adjoint on $L^2$, and its definition shows that
$$
Z[\beta,\varrho](\sigma)=Z_+-Z_-,
$$
where
$$
Z_{\pm}=\frac 1 {2i\pi} \int_{\gamma} e^{-\beta z }(z - H[\varrho])^{-1} n_\mp[\sigma] (z - H[\varrho])^{-1}dz \geq 0,
$$
for $n_\pm[\sigma]$ the positive and negative parts of $n[\sigma]$. Then, for $W=Z_++Z_- \geq 0$ with
$$
W=-\frac 1 {2i\pi} \int_{\gamma} e^{-\beta z }(z - H[\varrho])^{-1} |n[\sigma]| (z - H[\varrho])^{-1}dz,
$$
we have
$$
\Tr \left( (I+H_0)^{3/4} |Z[\beta,\varrho](\sigma)|(I+H_0)^{3/4} \right) \leq \Tr \left( (I+H_0)^{3/4} W(I+H_0)^{3/4} \right).
$$
The last inequality is simply a consequence of the triangle inequality in the space $\calJ^1(H^{3/2})$ when $H^{3/2}$ is equipped with the norm
$$
\| \varphi \|_{H^{3/2}}= \|(I+H_0)^{3/4} \varphi\|_{L^2}.
$$
In order to justify some calculations, we will work with the operator $Z_\eta$ that is introduced below: 
$$
Z_\eta:= R_\eta W R_\eta, \qquad R_\eta=(I+\eta H_0)^{-2}.
$$
Note first that the operator
$$
W_\eta:=(I+H_0)^{3/4}R_\eta(z - H[\varrho])^{-1} |n[\sigma]| (z - H[\varrho])^{-1}R_\eta (I+H_0)^{3/4}
$$
is trace class, uniformly in $z\in \gamma$, as, using \fref{belowresolv},
$$
\|W_\eta\|_{\calJ_1} \leq \|n[\varrho]\|_{L^\infty} \|(z-H[\varrho])^{-1}\|^2_{\calL(L^2)} \|(I+H_0)^{3/4} R_\eta\|_{\calJ_2}^2 \leq C_\eta.
$$
This will allow us to exchange operator traces and integrals over $z$. Since $R_\eta (I+H_0)^{3/4}$ is a bounded operator and $R_\eta$ and $H_0$ commute, we have
\begin{align*}
  \Tr\big(& (I+H_0)^{3/4} Z_\eta (I+H_0)^{3/4} \big) \\
  &= -\frac 1 {2i\pi} \int_{\gamma} e^{-\beta z } \Tr \left( R_\eta^2 (I+H_0) (z - H[\varrho])^{-1} |n[\sigma]| (z - H[\varrho])^{-1} \sqrt{I+H_0} \right) dz\\
&:= F_1+F_2,
\end{align*}
where
\bee
F_1&=&\frac 1 {2i\pi} \int_{\gamma} e^{-\beta z } \Tr \left( R_\eta^2 |n[\sigma]| (z - H[\varrho])^{-1} \sqrt{I+H_0} \right) dz\\
F_2&=&\frac 1 {2i\pi} \int_{\gamma} e^{-\beta z } \Tr \left( R_\eta^2 (A[\varrho]-z-1)(z - H[\varrho])^{-1}|n[\sigma]| (z - H[\varrho])^{-1} \sqrt{I+H_0} \right) dz.
\eee
For $u=\exp(-\beta H[\varrho])$ and exchanging again trace and integral, $F_1$ admits the simple expression
$$
F_1= \Tr \left( R_\eta^2 |n[\sigma]| u \sqrt{I+H_0}\right),
$$ 
that can be controlled by 
$$
|F_1| \leq \|n[\sigma]\|_{L^\infty} \| R^2_\eta\|_{\calL(L^2)} \| u \sqrt{I+H_0} \|_{\calJ_1}\leq C\|n[\sigma]\|_{L^\infty} \|u \|_{\calE}.
$$
With $\varrho \in B_M$, a similar estimate as \fref{estsigE} shows that $\|u\|_{\calE}$ can be bounded by some constant that depends on $M$, so that $|F_1| \leq C_M \|n[\sigma]\|_{L^\infty}$. Regarding $F_2$, we have
$$
|F_2| \leq C \int_{\gamma} dz |e^{-\beta z}| \| A[\varrho]-z-1\|_{L^\infty}\| n [\sigma]\|_{L^\infty} \|(z-H[\varrho])^{-1}\|_{\calJ_1}  \| (z-H[\varrho])^{-1} \sqrt{I+H_0} \|_{\calL(L^2)}.
$$
Following \fref{eq:bndeigen}, we have that $\lambda_m[\varrho]$ behaves like $m^2$ for large $m$. As a consequence, for any $z \in \gamma$ and any $\varrho \in B_M$,
$$
\|(z-H[\varrho])^{-1}\|_{\calJ_1}=\sum_{m \geq 0} |z-\lambda_m[\varrho]|^{-1} \leq C_M,
$$
where $C_M$ is independent of $z$. Also, since $z \in \gamma$, $\varrho \in B_M$ and $A[\varrho] \in L^\infty$, the operator $(z-H[\varrho])^{-1}$ is bounded  from $L^2$ to $\H2p$ uniformly in $z$ and in $\varrho$. This shows that $(z-H[\varrho])^{-1} \sqrt{I+H_0}$ is bounded in $\calL(L^2)$ independently of $z$ and $\varrho$ when $z \in \gamma$ and $\varrho \in B_M$. As a consequence, for any $\beta \in [\beta_0, \beta_1]$,
$$
|F_2| \leq C_M \|n[\sigma]\|_{L^\infty} \int_{\gamma} dz |e^{-\beta z}| (1+|z|) \leq C_{M,\beta_0,\beta_1} \|n[\sigma] \|_{L^\infty}.
$$ At this point, we have therefore proven that
$$
\Tr\left( (I+H_0)^{3/4} Z_\eta (I+H_0)^{3/4} \right) \leq C_{M,\beta_0,\beta_1} \| n[\sigma] \|_{L^\infty}.
$$ 
Since $Z_\eta$ converges to $W$ strongly in $\calL(L^2)$ as $\eta \to 0$, the above bound ends the proof of the lemma by semi-lower continuity of the trace norm.

Everything is now in place to study to the monotonicity of the kinetic energy.

\paragraph{Step 5: The kinetic energy is increasing.} We use the notation of Step 3, namely $\varrho_\eps$ is the unique minimizer of the penalized functional $F_\eps$ with $\varrho_\eps=\exp(-H_\eps/T)$, $H_\eps=H_0+A_\eps$. We will keep implicit the dependency of $\varrho_\eps$ and $A_\eps$ on $T$ to simplify notation, and introduce
\begin{equation*}
E_{\eps,T} : = E(\varrho_{\eps}).
\end{equation*}
We know from Step 4 that the derivative of $\varrho_{\eps}$ with respect to $T$ is continuous and belongs to $\calE$. We will need to regularize though to properly use the cyclicity of the trace and introduce for this the operator $R_\eta=(I+\eta H_0)^{-2}$. We have first
$$
E_{\eps,T}=\Tr \big(\sqrt{H_0} \varrho_\eps \sqrt{H_0} \big)= \lim_{\eta \to 0} \Tr \big(R_\eta \sqrt{H_0} \varrho_\eps \sqrt{H_0} R_\eta\big):=\lim_{\eta \to 0} E_\eta.
$$
Indeed, since $\partial_T \varrho$ is self-adjoint, we can decompose it into positive and negative parts, i.e. $\partial_T \varrho=B_+-B_-$, with $B_\pm \geq 0$, $B_+B_-=B_-B_+=0$, and $|\partial_T \varrho|=B_++B_-$. Since $B_\pm \in \calE^+$ for all $T>0$, it then follows easily that
$$
\Tr \big(\sqrt{H_0} B_\pm \sqrt{H_0} \big)= \lim_{\eta \to 0} \Tr \big(R_\eta \sqrt{H_0} B_\pm \sqrt{H_0} R_\eta\big).
$$
Using the cyclicity of the trace and the fact that $R_\eta$ and $H_0$ commute, we find
\begin{align*}
\partial_T E_{\eta} &= \Tr\left(R_\eta^2 H_0  \partial_T \varrho_{\eps} \right) = \Tr\left( R_\eta^2H_{\eps}  \partial_T \varrho_{\eps}\right) - \Tr\left(R_\eta^2 A_{\eps}  \partial_T\varrho_{\varepsilon}\right)
\\ &=  \Tr\left( R_\eta^2\left( H_{\eps} - T \partial_TA_{\eps} \right)   \partial_T \varrho_{\varepsilon}\right)  - \Tr\left( R_\eta^2\left( A_{\eps} -  T\partial_TA_{\eps} \right)\partial_T\varrho_{\varepsilon}\right).
\end{align*}
Above, $\eps \partial_TA_{\eps}=n[\partial_T \varrho_\eps]$ is well-defined in $L^\infty$ thanks to \fref{ninfty} since $\partial_T \varrho_\eps \in \calE$. Our first task is to pass to the limit $\eta \to 0$ to recover $\partial_T E_{\eps,T}$. Since $A_\eps$ and $\partial_T A_\eps$ belong to $L^\infty$ and $\partial_T \varrho_\eps$ is trace class, the only term requiring attention is the first one in the r.h.s. above. For this, and following the expression of the differential of the functional $\calI$ defined in Step 4, we have  
\begin{align*}
\partial_T \varrho_{\varepsilon} &= \frac 1 {T^2} H_{\eps}\varrho_{\varepsilon} - \frac 1 {2\pi i}\int_{\gamma} e^{- \frac zT} (z- H_{\eps})^{-1} \partial_T A_{\eps}(z- H_{\eps})^{-1}dz
\\ &= :  G_{1,\varepsilon}  + G_{2,\varepsilon}.
\end{align*}
The first term $G_{1,\eps}$ is such that $H_\eps  G_{1,\varepsilon} \in \calJ_1$ and therefore poses no issue in the limit $\eta \to 0$. Regarding the second one, we write
$$
H_\eps  G_{2,\varepsilon}=-\frac 1 {2\pi i }\int_{\gamma} e^{- \frac zT} \left(\calG_1(z)+\calG_2(z) \right)dz,
$$
with
\bee
\calG_1(z) &=&-\partial_T A_{\eps}(z- H_{\eps})^{-1}\\
\calG_2(z) &=&z(z- H_{\eps})^{-1}\partial_T A_{\eps}(z- H_{\eps})^{-1}.
\eee
Both $\calG_1$ and $\calG_2$ are trace class since $\partial_T A_\eps \in L^\infty$ and $(z- H_{\eps})^{-1}$ is trace class as we saw in the proof of Lemma \ref{boundZ}. Besides, $\calG_1$ is uniformly bounded in $\calJ_1$ w.r.t. $z$ and $\calG_2$ is bounded by $|z|$ in $\calJ_1$. This shows that $H_\eps  G_{2,\varepsilon} \in \calJ_1$, and as a consequence we can pass to the limit $\eta \to 0$ and obtain that 
\begin{align*}
\partial_T E_{\eps,T} &=  \Tr\left(\left( H_{\eps} - T \partial_TA_{\eps} \right)   \partial_T \varrho_{\varepsilon}\right)  - \Tr\left(\left( A_{\eps} -  T\partial_TA_{\eps} \right)\partial_T\varrho_{\varepsilon}\right).
\end{align*}
By using
\begin{align*}
&\Tr\big( A_{\eps}  \partial_T\varrho_{\varepsilon}\big) = \big(A_{\eps}, n[\partial_T\varrho_{\varepsilon}]\big) = \varepsilon \partial_T \| A_{\eps}\|^2_{L^2}\\
&\Tr\big(\partial_T A_{\eps}  \partial_T\varrho_{\varepsilon}\big) = \big(\partial_T A_{\eps} ,n[\partial_T\varrho_{\varepsilon}]\big) = \varepsilon \|\partial_T A_{\eps}\|^2_{L^2},
\end{align*}
we deduce that
\begin{align*} 
\partial_T E_{\eps,T} &+\varepsilon \partial_T \| A_{\eps}\|^2_{L^2} -\varepsilon T \|\partial_T A_{\eps}\|^2_{L^2}\\
&= \Tr\left( H_{\eps}   G_{1,\varepsilon} \right) + \Tr\left( H_{\eps}   G_{2,\varepsilon } \right) -T\Tr\left( \partial_T A_{\eps}  G_{1,\varepsilon} \right) - T\Tr\left(\partial_T A_{\eps}  G_{2,\varepsilon} \right) \nonumber
\\  &= : \calT_{1,\varepsilon} +  \calT_{2,\varepsilon} + \calT_{3,\varepsilon}  + \calT_{4,\varepsilon} .
\end{align*}
The key observation is that $\calT_{2,\varepsilon} + \calT_{3,\varepsilon}=0$, while  $\calT_{1,\varepsilon} + \calT_{4,\varepsilon} \geq 0$. Indeed, with the definition of $\calG_1$ and $\calG_2$ above, we have first that
$$
-\frac 1 {2\pi i} \Tr \left(\int_{\gamma} e^{- \frac zT} \calG_1(z)dz \right)=\Tr \big( \partial_T A_\eps \varrho_\eps \big).
$$
On the other hand, since we have seen that $\calG_2$ has sufficient regularity, we can exchange trace and integral to arrive at
\bea \nonumber
-\frac 1 {2\pi i} \Tr \left(\int_{\gamma} e^{- \frac zT} \calG_2(z)dz \right)&=&-\frac 1 {2\pi i} \int_{\gamma} z e^{- \frac zT} \Tr \left((z- H_{\eps})^{-1}\partial_TA_{\eps}(z- H_{\eps})^{-1} \right)dz\\ \nonumber
&=&-\frac 1 {2\pi i} \int_{\gamma} z e^{- \frac zT} \Tr \left(\partial_TA_{\eps}(z- H_{\eps})^{-2} \right)dz\\
&=&-\frac 1 {2\pi i} \Tr \left(\partial_TA_{\eps} \int_{\gamma} z e^{- \frac zT} (z- H_{\eps})^{-2}dz \right). \label{calc}
\eea
If $f(x)=x e^{-x/T}$, then standard complex analysis shows that the complex integral above is just $f'(H_\eps)$, that is $\varrho_\eps-H_\eps \varrho_\eps /T$. This means that
$$
\Tr \big( H_\eps G_{2,\eps} \big)=\frac{1}{T}\Tr \big( \partial_TA_\eps H_\eps \varrho_\eps \big)=T \Tr\big(\partial_T A_\eps G_{1,\eps} \big),
$$
which proves the claim that $\calT_{2,\varepsilon} + \calT_{3,\varepsilon}=0$. The fact that $T^2 \calT_{1,\varepsilon} = \Tr\left(H_{\eps}^2 \varrho_{\varepsilon} \right) \geq 0$ is clear. The fact that $\calT_{4,\varepsilon} \geq 0$ is an easy consequence of \fref{negZ} with $\sigma=\partial_T \varrho_\eps$. At this stage, we have therefore arrived at 
\begin{align} \label{leqE}
\partial_T E_{\eps,T} \geq - \varepsilon \partial_T \| A_{\eps}\|^2_{L^2}.
\end{align}
Integrating the previous inequality and expliciting the dependency of $A_\eps$ on $T$, we find, $\forall T_2\geq T_1>0$,
$$
E_{\eps,T_2} - E_{\varepsilon, T_1}\geq -\varepsilon\left(\| A_{\eps,T_2}\|^2_{L^2} - \| A_{\eps,T_1}\|^2_{L^2}\right).
$$
As claimed at the beginning of Step 3, we have, for all $T>0$ and all $n_0\in \Hp$,
\begin{equation*}
E_{\eps,T} \underset{\varepsilon\rightarrow 0}{\longrightarrow}E(\varrho_{T,n_0}),
\end{equation*}
which, thanks to the crucial uniform bound \eqref{unifA}, leads to
\begin{equation}\label{ineq:NRJ}
E(\varrho_{T_2,n_0}) \geq E(\varrho_{T_1,n_0}), \qquad \forall T_2 \geq T_1>0.
\end{equation}

Note that there is no information about $\partial_T E(\varrho_{T,n_0})$ since we do not have the uniform estimates required to pass to the limit in \fref{leqE}. We derive in the next step a relation between the energy and the entropy that will allows us to prove that the energy is strictly increasing and the entropy strictly decreasing. 

\paragraph{Step 6: Relation between the derivatives of the kinetic energy and the entropy, and decrease of the entropy.} For any $T>0$ and $\varrho_{\eps,T}$ the unique minimizer of the penalized problem (we make explicit here the dependency of $\varrho_{\eps}$, $F_\eps$ and $A_\eps$ w.r.t. $T$ since it will be needed further), the goal of this section is to prove the following relation, 
\be \label{relE}
\frac{d }{dT}E(\varrho_{\eps,T})= T \frac{d }{dT}S(\varrho_{\eps,T})+\eps \frac{d }{dT}\|A_{\eps,T}\|_{L^2}^2,
\ee
or more exactly an integrated version of it that will be more useful when passing to the limit $\eps \to 0$. We will see that expression \fref{relE} is a consequence of standard calculus of variations arguments. Consider first the penalized free energy $F_{\eps,T}(\varrho_{\eps,T})$ that we aim to differentiate w.r.t. $T$. Since we have proved in Step 4 that for any $T>0$, $\varrho_{\eps,T}$ is continuously differentiable w.r.t. $T$ with values in $\calE$, it follows that $E(\varrho_{\eps,T})$ and $\|n[\varrho_{\eps,T}]-n_0\|_{L^2}^2$ are differentiable, and we have
\begin{align*}
&\frac{d}{dT} E(\varrho_{\eps,T})= \Tr \big( \sqrt{H_0} \partial_T \varrho_\eps\sqrt{H_0}\big) < \infty\\
&\frac{1}{2 \eps}\frac{d}{dT} \| n[\varrho_{\eps,T}]-n_0 \|^2_{L^2}= \frac{1}{\eps}(n[\partial_T \varrho_\eps], n[\varrho_\eps]-n_0)<\infty. 
\end{align*}
In particular, the last term can be recast as, using the definition of $A_\eps$,
$$ \frac{1}{2 \eps}\frac{d}{dT} \| n[\varrho_{\eps,T}]-n_0 \|^2_{L^2}=\frac{\eps}{2} \frac{d}{dT}  \|A_{\eps,T}\|^2_{L^2}.
$$
We have the following lemma, whose proof is postponed to the end of the section:
\begin{lemma} \label{derivS} The derivative of $S(\varrho_{\eps,T})$ with respect to $T$ is well-defined for all $T>0$.
\end{lemma}
The above lemma shows that the G\^ateaux derivative of $S$ at $\varrho_{\eps,T}$ in the direction $\partial_T \varrho_{\eps,T}$ is well-defined, as a consequence so does that of $F_{\eps,T}$. We are now ready to conclude. We have first on the one hand,
\be \label{dif1}
\frac{d}{dT} F_{\eps,T}(\varrho_{\eps,T})=\frac{d}{dT} E(\varrho_{\eps,T})+T \frac{d}{dT} S(\varrho_{\eps,T})+S(\varrho_{\eps,T})+\frac{1}{2}\frac{d}{dT} \| n[\varrho_{\eps,T}]-n_0 \|^2_{L^2}.
\ee
On the other hand, since $\varrho_{\eps,T}$ is the minimizer of $F_{\eps,T}$, we have $DF_{\eps,T}[\varrho_{\eps,T}](\partial_T \varrho_{\eps,T})=0$, and as a consequence
\be \label{dif2}
\frac{d}{dT} F_\eps(\varrho_{\eps,T})=DF_\eps[\varrho_{\eps,T}](\partial_T \varrho_{\eps,T})+S(\varrho_{\eps,T})=S(\varrho_{\eps,T}).
\ee
Combining \fref{dif2} with \fref{dif1}, we find \fref{relE}. After integration, the latter becomes
$$
E(\varrho_{\eps,T_2})-E(\varrho_{\eps,T_1})=T_2 S(\varrho_{\eps,T_2})-T_1 S(\varrho_{\eps,T_1})-\int_{T_1}^{T_2} S(\varrho_{\eps,T}) dT+\eps \|A_{\eps,T_2}\|^2_{L^2}-\eps \|A_{\eps,T_1}\|^2_{L^2}.
$$
We have already mentioned that $\varrho_{\eps,T}$ converges to $\varrho_{T,n_0}$ strongly in $\calE$ for all $T>0$, so that $E(\varrho_{\eps,T})$ converges to $E(\varrho_{\eps,T})$ for all $T>0$. An application of Lemma \ref{propentropie} (ii) shows that $S(\varrho_{\eps,T})$ converges to $S(\varrho_{\eps,T})$. We then pass to the limit in the integral using dominated convergence and estimate \fref{estiml} further. Using \fref{unifA}, we arrive at, for all $T_2,T_1>0$,
\be \label{eq:relES0}
E(\varrho_{T_2,n_0})-E(\varrho_{T_1,n_0})=T_2 S(\varrho_{T_2,n_0})-T_1 S(\varrho_{T_1,n_0})-\int_{T_1}^{T_2} S(\varrho_{T,n_0}) dT.
\ee
This is one of the main results of this section. The other one concerns the fact that the entropy is nonincreasing. Indeed, combining \fref{relE} and \fref{leqE}, we arrive at
$$
\frac{d }{dT}S(\varrho_{\eps,T}) \geq -2 \eps \frac{1}{T} \frac{d }{dT}\|A_{\eps,T}\|_{L^2}^2.
$$  
Integrating, we find, for all $T_2  \geq T_1 >0$,
$$
S(\varrho_{\eps,T_2})-S(\varrho_{\eps,T_1}) \geq - 2 \eps \left[ \frac{\|A_{\eps,T}\|_{L^2}^2}{T} \right]_{T_1}^{T_2} - 2 \eps \int_{T_1}^{T_2} \frac{\|A_{\eps,T}\|_{L^2}^2}{T^2}dT.
$$
Sending $\eps$ to zero with estimate \fref{unifA} yields the desired result.\\

We conclude with the proof of Lemma \ref{derivS}.\\

\noindent \textit{Proof of Lemma \ref{derivS}}. We need to differentiate the entropy term, and have to be a little careful since $\log x$ has a singularity at $x=0$. For $\eta>0$, we hence regularize the entropy as
$$
s_\eta(x) = (x+ \eta) \log(x+\eta) -x - \eta \log \eta.
$$
For any $\varrho \in \calE^+$, the operator $\log(\varrho+\eta)$ is bounded. Then, the G\^ateaux derivative of $S_\eta(\varrho):=\Tr (s_\eta (\varrho))$ in the direction $u \in \calE^+$ is well-defined and reads, see \cite[Lemma 5.3]{mehats2011problem} for a proof, 
$$
D S_\eta[\varrho](u)= \Tr \big(\log(\eta+\varrho) u \big).
$$
Hence, since $\partial_T \varrho_{\eps,T} \in \calE$,
$$
\frac{d}{dT}S_\eta(\varrho_{\eps,T})=-\Tr \big(\log(\eta+\varrho_{\eps,T}) \partial_T \varrho_{\eps,T} \big).
$$
We sent now $\eta$ to zero by observing that, for $\beta=1/T$,
\begin{align*}
\Tr \big(\log(\eta&+\varrho_{\eps,T}) \partial_\beta \varrho_{\eps,T} \big)\\
&=-\Tr \big(\log(\eta+\varrho_{\eps,T}) H_{\eps,T} \varrho_{\eps,T} \big)+\Tr \big(\log(\eta+\varrho_{\eps,T}) Z[\beta,\varrho_{\eps,T}](\partial_\beta\varrho_{\eps,T}) \big)\\
&:=G_{1.\eta}(T)+G_{2,\eta}(T),
\end{align*}
where $H_{\eps,T}=H_0+A_{\eps,T}$ equipped with $D(H_{\eps,T})=\H2p$, and $Z[\beta,\varrho](\sigma)$ is defined in Step 4b. The first term $G_{1,\eta}$ is just
$$
G_{1,\eta}(T)=-\sum_{p\in \Nm} \log(\eta+e^{-\beta \lambda_{p,\eps}}) \lambda_{p,\eps} e^{-\beta \lambda_{p,\eps}}, 
$$
where the $\lambda_{p,\eps}$ are the eigenvalues of $H_{\eps,T}$. Moreover, calculations close to \fref{calc} using the cyclicity of the trace and the fact that $\log(\eta+\varrho_{\eps,T})$ and $(z-H_{\eps,T})^{-1}$ commute show that the second term $G_{2,\eta}$ is equal to
$$
G_{2,\eta}(T)=\Tr \big(\log(\eta+\varrho_{\eps,T}) \varrho_{\eps,T} n[\partial_T \varrho_{\eps,T}] \big).
$$
At that point, we therefore have obtained that, for all $T_1,T_2>0$,
$$S_\eta(\varrho_{\eps,T_2})-S_\eta(\varrho_{\eps,T_1})=\int_{T_1}^{T_2} \left(-G_{1,\eta}(T)/T^2+G_{2,\eta}(T)  \right) dT,
$$
and we need some estimates to pass to the limit. The term $G_{1,\eta}$ is treated directly using the min-max principle and estimate \fref{eq:bndeigen}. Its limit is denoted by $G_1$ and satisfies $\sup_{T \in [T_1,T_2]} |G_1| \leq C$. For the other terms, a similar calculation as \fref{elog} with the $L^2$ norm replaced by the $L^1$ norm, shows that, together with \fref{unifT}, $\forall a \in [0,1]$,
\be \label{estiml}
\|\log(a+\varrho_{\eps,T}) \varrho_{\eps,T} \|_{\calJ_1} \leq C+C \|\varrho_{\eps,T}\|_{\calE} \leq C+C T^2,
\ee
where $C$ does not depend on $a$. We then claim that $\log(\eta+\varrho_{\eps,T}) \varrho_{\eps,T}$ converges to $\log(\varrho_{\eps,T}) \varrho_{\eps,T}$ strongly in $\calJ_1$ for all $T>0$. Indeed, on the one hand, for $(\rho_p)_{p\in \Nm}$ the eigenvalues of $\varrho_{\eps,T}$,
\be \label{CVJ1}
\|\log(\eta+\varrho_{\eps,T}) \varrho_{\eps,T}\|_{\calJ_1}=\sum_{p \in \Nm} \rho_p|\log (\eta+\rho_p)| \to \sum_{p \in \Nm}\rho_p|\log (\rho_p)| =\|\log(\varrho_{\eps,T}) \varrho_{\eps,T}\|_{\calJ_1},
\ee
while, on the other, a similar calculation as above shows that  $\log(\eta+\varrho_{\eps,T}) \varrho_{\eps,T}$ converges to $\log(\varrho_{\eps,T}) \varrho_{\eps,T}$ strongly in $\calL(L^2)$. Lemma \ref{theoSimon} then yields the desired result. As a consequence, $\forall T>0$,
$$
\lim_{\eta \to 0} S_\eta(\varrho_{\eps,T})=S(\varrho_{\eps,T}), \qquad \lim_{\eta \to 0} G_{2,\eta}=G_2=\Tr \big(\log(\varrho_{\eps,T}) \varrho_{\eps,T} n[\partial_T \varrho_{\eps,T}] \big).
$$
Finally, the fact that $\varrho_{\eps,T} \in \calC^1(\Rm^*_+,\calE)$, together with \fref{estiml}-\fref{CVJ1} yields
\bee
\sup_{T \in [T_1,T_2]}|G_{2,\eta}| &\leq& \sup_{T \in [T_1,T_2]} \|\log(\eta+\varrho_{\eps,T}) \varrho_{\eps,T}\|_{\calJ_1} \sup_{T \in [T_1,T_2]}\|n[\partial_T \varrho_{\eps,T}]\|_{L^\infty} \\
&\leq& C \sup_{T \in [T_1,T_2]}\|\partial_T \varrho_{\eps,T}\|_{\calE} \leq C.
\eee
which allows us to use dominated convergence and obtain
$$S(\varrho_{\eps,T_2})-S(\varrho_{\eps,T_1})=\int_{T_1}^{T_2} ( G_1+G_2)dT, \qquad \forall T_2 \geq T_1>0.
$$
Since $\varrho_{\eps,T}$ is continuous w.r.t. $T$ in $\calE$, it can be shown with similar proofs as above that $G_1$ and $G_2$ are continuous in $T$ as well, we omit the details. This then concludes the proof of the lemma since $T_1$ and $T_2$ are arbitrary.

\paragraph{Step 7: $\mathsf{E}_T$ and $\mathsf{S}_T$ are strictly increasing.}
The aim of this step is to prove that we have the strict inequality
\begin{equation*}
\mathsf{E}_{T_1}> \mathsf{E}_{T_0}, \qquad \forall T_1>T_0 >0. 
\end{equation*}
In Step 5, we have proved the non strict inequality \eqref{ineq:NRJ}, and it therefore only remains to prove that if $T_1\neq T_0$, then
\begin{equation*}
\mathsf{E}_{T_1} \neq  \mathsf{E}_{T_0}.
\end{equation*}
Since the energy is nondecreasing and continuous, we proceed by contradiction and assume that there exists a non empty open subset $I = (T_0,T_1)$ of $\mathbb{R}^{*}_+$ such that $\mathsf{E}_u =  \mathsf{E}_v$, $\forall u,v\in I.$
From \eqref{eq:relES0}, it follows that 
\be \label{eq:relES}
u S(\varrho_{u,n_0})-v S(\varrho_{v,n_0})=\int_{v}^{u} S(\varrho_{T,n_0}) dT, \qquad \forall u,v\in I.
\ee
Since we have seen in Step 1 that $\mathsf{S}_T$ is continuous on $\Rm^*_+$, this shows that $(u \mathsf{S}_u)'$ exists, is continuous and equal to $ \mathsf{S}_u$ for all $u \in I$. As a consequence, $\mathsf{S}_u' = 0$ for all $u \in I$. We then deduce that, $\forall u,v\in I$,
\begin{equation*}
F_u(\varrho_{u,n_0}) = \mathsf{E}_u + u\mathsf{S}_u = \mathsf{E}_v + u \mathsf{E}_v = F_u(\varrho_{v,n_0}).
\end{equation*}
Since the minimizer of $F_u$ is unique, we have $\varrho_{u,n_0} = \varrho_{v,n_0}$ which, by \eqref{eq:solminim1D}, leads to the equality
\begin{equation}\label{eq:Hequv}
\frac{H[\varrho_{u,n_0}]}{u} = \frac{H[\varrho_{v,n_0}]}{v}, \qquad H[\varrho_{v,n_0}]=H_0+A[\varrho_{v,n_0}].
\end{equation}
 For $A[\varrho_{v,n_0}]$, expression \fref{expA} becomes 
$$
A[\varrho_{v,n_0}]= \frac{1}{n_0} \left(\frac{1}{4} \Delta n_0 + \frac{1}{2} n[\nabla \varrho_{v,n_0} \nabla ] + v n[\varrho_{v,n_0} \log \varrho_{v,n_0}] \right).
$$
Since $\varrho_{u,n_0} = \varrho_{v,n_0}$, we can write
$$
A[\varrho_{v,n_0}]= f+ u g, \qquad A[\varrho_{v,n_0}]= f+ v g, \qquad f,g \in L^2.
$$
It finally follows from \eqref{eq:Hequv},
$$
\frac{H_0+f}{u} = \frac{H_0+f }{v},
$$
which implies $u = v$ for all $u,v\in I$. This is in contradiction with the fact that $I$ is open. Hence, $\mathsf{E}_T$ is strictly increasing. The fact that $\mathsf{S}_T$ is strictly decreasing is now straightforward. Indeed, we already know from Step 6 that $\mathsf{S}_T$ is nonincreasing. Supposing that it is constant for all $u,v \in I$ for some open interval $I$, we find from \fref{eq:relES0} that $\mathsf{E}_u=\mathsf{E}_v$, for all $u,v \in I$, which contradicts the fact that $\mathsf{E}_T$ is strictly increasing. Hence, $\mathsf{S}_T$ is strictly decreasing. This ends the proof of Theorem \ref{thm:NRJIncr}.
  
\section{Appendix} \label{appendix}
We give here a few technical lemmas that are used throughout the paper. The simple proof of the first one can be found in \cite[Lemma 5.3]{mehats2017quantum}.  

\begin{lemma} \label{lieb} Suppose $\varrho$ is self-adjoint and belongs to $\calE$. Then, the following estimates hold:
\begin{align}
\label{ninfty}\|n[\varrho]\|_{L^\infty} \leq C \| \varrho\|^{1/4}_{\calJ_2}  \|\varrho\|_{\calE}^{3/4}\\
\label{gradnl2}\|\nabla n[\varrho]\|_{L^2} \leq C \| \varrho\|^{1/4}_{\calJ_1}  \|\varrho\|_{\calE}^{3/4}.
\end{align}
\end{lemma}
The second lemma provides us with important compactness results.

\begin{lemma}[\cite{MP-JSP}, Lemma 3.1]\label{strongconv} 
(i) Let $(\varrho_k)_{k\in \Nm}$ be  a bounded sequence of $\calE^+$. Then, up to an extraction of a subsequence, there exists $\varrho\in \calE^+$ such that
$$
\varrho_k\to\varrho\mbox{ in }\calJ_1\quad \mbox{as } k\to +\infty
$$
and
$$
\Tr \big(\sqrt{H_0}\varrho\sqrt{H_0}\big)\leq \liminf_{k\to +\infty} \Tr \big(\sqrt{H_0}\varrho_k\sqrt{H_0}\big).
$$
(ii) Let $(\varrho_k)_{k\in \Nm}$ be a sequence of operators such that $\varrho_k \to \varrho$ strongly in $\calL(L^2)$ and $\Tr \big((I+H_0)^{\alpha}|\varrho_k| (I+H_0)^{\alpha} \big)\leq C$ for some $\alpha>1/2$. Then, $\varrho_k$ converge to $\varrho$ strongly in $\calE$. 
\end{lemma}

Item (ii) is not stated in \cite{MP-JSP} and we detail here a proof for completeness. First, there exists $B$ and a subsequence such that $(I+H_0)^{\alpha}|\varrho_k| (I+H_0)^{\alpha} $ converges to $B$ weakly-$*$ in $\calJ_1$. Since $|\varrho_k| \to |\varrho|$ in $\calL(L^2)$ as $\varrho_k \to \varrho $ in $\calL(L^2)$, and $(I+H_0)^{\alpha}$ is injective, this implies that $B=(I+H_0)^{\alpha}|\varrho| (I+H_0)^{\alpha}$. Write then, since $(I+H_0)^{1-2\alpha}$ is a compact operator on $L^2$,
\bea \nonumber
\Tr \big((I+H_0)^{1/2}|\varrho_k| (I+H_0)^{1/2} \big)&=&\Tr \big((I+H_0)^{\alpha}|\varrho_k| (I+H_0)^{\alpha} (I+H_0)^{1-2\alpha} \big)\\
&\to &\Tr \big((I+H_0)^{\alpha}|\varrho| (I+H_0)^{\alpha} (I+H_0)^{1-2\alpha} \big) \label{conH}\\
&=&\Tr \big((I+H_0)^{1/2}|\varrho| (I+H_0)^{1/2} \big). \nonumber
\eea
When $H^1$ is equipped with the norm
$$
\|u\|_{H^1}=\| (I+H_0)^{1/2} u\|_{L^2},
$$
we have the identification
\be
\Tr \big((I+H_0)^{1/2}|\varrho| (I+H_0)^{1/2} \big)=\|\varrho \|_{\calJ_1(H^1)}=\|\varrho \|_{\calE}. \label{ident}
\ee
In order to apply Lemma \ref{theoSimon}, it remains to prove that $\varrho_k$ converges weakly to $\varrho$ in the sense of bounded operator in $H^1$, that is
$$
((I+H_0)^{1/2} \varrho_k u, (I+H_0)^{1/2}v) \to ((I+H_0)^{1/2} \varrho u, (I+H_0)^{1/2} v), \qquad \forall u,v \in H^1.
$$
From the observation that $\varrho_k$ and $\varrho$ are uniformly bounded in $\calL(H^1)$ according to \fref{ident} and \fref{conH}, we can actually consider test functions $v$ above in $H^2$. The result then follows since $(I+H_0)^{1/2}$ is self-adjoint and $\varrho_n \to \varrho$ in $\calL(L^2)$. Lemma \ref{theoSimon} finally yields the strong convergence of $\varrho_n$ to $\varrho$ in $\calJ_1(H^1)$ and therefore in $\calE$.

The lemma below provides us with important properties of the entropy.
\begin{lemma}[\cite{MP-JSP}, Lemma 5.2]
\label{propentropie}
The application $\varrho\mapsto \Tr (\varrho \log \varrho-\varrho)$ possesses the following properties.\\
(i) There exists a constant $C>0$ such that, for all $\varrho\in \calE^+$, we have
\be
\label{souslin}
\Tr \big(\varrho\log \varrho-\varrho\big)\geq -C\left(\Tr \big(\sqrt{H_0}\varrho\sqrt{H_0}\big)\right)^{1/2}.
\ee
(ii) Let $\varrho_k$ be a bounded sequence of $\calE^+$ such that $\varrho_k$ converges to $\varrho$ in $\calJ_1$, then $\varrho_k \log \varrho_k-\varrho_k$ converges to $\varrho \log \varrho-\varrho$ in $\calJ_1$.\\
\end{lemma}

The next two lemmas state some regularity results for density operators. The second one  is adapted from \cite{pinaud2016quantum} to account for the dependency in $T$.
\begin{lemma} [\cite{MP-JSP}, Lemma A.1] \label{lieb2} Let $\varrho \in \calE^+$. Then, we have the estimate
$$
\Tr \big( \varrho ^{2/3}\big) \leq C \left(\Tr \big( \sqrt{H_0} \varrho \sqrt{H_0}\big)\right)^{2/3}.
$$
\end{lemma}
\begin{lemma} [\cite{pinaud2016quantum}, Lemma 4.4] \label{regmin} Let  $V \in L^2(\Omega)$  and define $\varrho=\exp(-(H_0+V)/T)$, where $D(H_0+V)=\H2p$. Then $\varrho \in \calE^+$ and $H_0 \varrho H_0 \in \calJ_1$, with the estimate
  $$
\Tr\big( H_0 \varrho H_0 \big) \leq C(1+T^2) \left(
1+ \big(1+\|V\|^2_{L^2}\big) \|\varrho\|_\calE \right).
$$
\end{lemma}

The last lemma is a technical result. 
\begin{lemma}\label{lem:bndvarsigma}
Let $H = H_0 + A$, with $A\in L^{\infty}$ and domain $\H2p$. Denote by $(\lambda_k[H])_{k \in \Nm}$ the nondecreasing sequence of eigenvalues of $H$ counted with multiplicity. Then,
\begin{equation*}
\sum_{m,k = 0}^{+\infty}-\varsigma_{m,k}[H]<+\infty,
\end{equation*}
where
\begin{equation*}
\varsigma_{m,k}[H] = \left\{\begin{array}{ll}
 - e^{- \lambda_k[H]},\,\mbox{ if } \lambda_k[H]= \lambda_m[H],
 \\ \frac{e^{- \lambda_k[H]} - e^{- \lambda_m[H]}}{\lambda_k[H] - \lambda_m[H]},\,\mbox{ if } \lambda_k[H]\neq \lambda_m[H],
\end{array}\right.
\end{equation*}

\end{lemma}
\begin{proof}
We first remark that, by symmetry,
\begin{equation*}
\sum_{m,k = 0}^{+\infty}-\varsigma_{m,k}[H] =\sum_{m = 0}^{+\infty}- \varsigma_{m,m}[H] +  2 \sum_{k = 0}^{+\infty} \sum_{m = 0}^{k-1} -\varsigma_{m,k}[H] = : I_1 + I_2+I_3, 
\end{equation*}
where
$$
I_2=\sum_{k = 0}^{+\infty}\sum_{m = 0, \lambda_m[H] = \lambda_k[H]}^{k-1}-\varsigma_{m,k}[H], \qquad I_3=\sum_{k = 0}^{+\infty}\sum_{m = 0, \lambda_{m}[H] \neq \lambda_{k}[H]}^{k-1}-\varsigma_{m,k}[H].
$$
The min-max principle yields
\begin{equation}\label{eq:minmaxprp}
\gamma_m - \|A\|_{L^{\infty}} \leq \lambda_m[H] \leq \gamma_m + \|A\|_{L^{\infty}},
\end{equation}
where the $\gamma_m $ are the eigenvalues of $H_0$, which, counting muliplicities, read $\gamma_0=0$, $\gamma_{2k} =\gamma_{2k-1}= 2 (\pi k)^2$ for $k \geq 1$.
Concerning $I_1$, we then immediately have
\begin{equation*}
I_1 = \sum_{m = 0}^{+\infty} e^{-\lambda_m[H]} < +\infty.
\end{equation*}
We now turn to $I_2$ and $I_3$. From \eqref{eq:minmaxprp} and the expression of $\gamma_m$,  we can see that, for $n \geq 1$ and $m=2p$,
\bee
\lambda_{m+n}[H] - \lambda_{m}[H]&\geq& \gamma_{2p+n}-\gamma_{2p} - 2 \|A\|_{L^\infty}\\
&\geq& 2\pi^2 \left(n p + \frac{n^2}{4} -\|A\|_{L^{\infty}}/\pi^2\right)\\
&\geq &2\pi^2 \left(p + \frac{1}{4} -\|A\|_{L^{\infty}}/\pi^2\right).
\eee
Denoting by $\ell_0$ the integer part of $\|A\|_{L^{\infty}}/(2\pi^2)$, the r.h.s. above is positive provided $p \geq \ell_0$. When $m=2p-1$, a similar estimate as above holds now for $n\geq 2$ due to the multiplicity of $\gamma_m$. Thus, for $k \geq 2(\ell_0+1)$, the indices $m \leq k-1$ such that $\lambda_m[H]=\lambda_k[H]$ are at most $m=k-1$. Hence,
$$
\sum_{k = 2(\ell_0+1) }^{+\infty}\sum_{m = 0, \lambda_m[H] = \lambda_k[H]}^{k-1}-\varsigma_{m,k}[H]  \leq \sum_{k = 2(\ell_0+1) }^{+\infty} - e^{-\lambda_{k-1}[H]} < \infty,
$$
which shows that $I_2$ is finite. Let $N=\min\left\{n\in\mathbb{N}:\, \lambda_{n}[H] \geq 0\right\}$ and $N_0=\max(N,2(\ell_0+1))$. Consider now
$$
I_4=\sum_{k = N_0}^{+\infty}\sum_{m = N_0}^{k-2}-\varsigma_{m,k}[H], \qquad I_5=\sum_{k = N_0}^{+\infty}\sum_{m = 0}^{N_0-1}-\varsigma_{m,k}[H].
$$
For $k>m\geq N$ such that $\lambda_m[H] < \lambda_k[H]$,
\begin{align*}
-\varsigma_{m,k}[H] \leq \frac 1 {\lambda_k[H]}\frac{e^{-\frac{\lambda_m[H]}T}}{1-\lambda_m[H]/\lambda_k[H]}\leq  \frac 1 {\lambda_k[H]} \frac{\lambda_{m+\ell(m)}[H]e^{-\frac{\lambda_m[H]}T}}{\lambda_{m+\ell(m)}[H]-\lambda_{m}[H]},
\end{align*}
where 
\begin{equation*}
\ell(m) = \min\left\{n\in\mathbb{N}: \, \lambda_{m+n}[H]>\lambda_m[H] \right\}.
\end{equation*}
According to the discussion above, $\ell(2p) = 1$ and $\ell(2p+1)=2$ for $p \geq \ell_0$. Hence, 
\bee
I_4&\leq& \sum_{k = N_0}^{+\infty} \frac 1 {\lambda_k[H]} \sum_{m =N_0}^{\infty}\frac{\lambda_{m + \ell(m)}[H]e^{-\frac{\lambda_m[H]}T}}{\lambda_{m+\ell(m)}[H]-\lambda_{m}[H]}<+\infty.
\eee
Finally, 
$$
I_5 \leq \sum_{k=0}^\infty \sum_{m=0}^{N_0-1} \frac{C}{\lambda_{k+N_0}-\lambda_{m}} \leq \sum_{k'=0}^\infty \sum_{k=0}^\infty \frac{C N_0 }{\lambda_{k+N_0}-\lambda_{N_0-1}} \leq \sum_{k=0}^\infty \frac{C}{1+k^2}<\infty,
$$
which shows that $I_5$ is finite and ends the proof.
\end{proof}

{\footnotesize \bibliographystyle{siam}
  \bibliography{bibliography.bib} }

 \end{document}